\documentclass[10pt]{article}

\usepackage{color}
\definecolor{DarkBlue}{rgb} {0,0,0.5}

\usepackage{amsfonts}
\usepackage{amssymb}
\usepackage{amsmath} 
\usepackage{amsthm}

\newcommand{\Tr}{{\mathrm{Tr}}}

\usepackage{fullpage} 

\usepackage[bookmarksnumbered,colorlinks,linkcolor=DarkBlue,citecolor=DarkBlue,urlcolor=DarkBlue]{hyperref} 

\newtheorem{thm}{Theorem}
\newtheorem{clm}[thm]{Claim}
\newtheorem{lem}[thm]{Lemma}
\newtheorem{cor}[thm]{Corollary}

\newcommand{\cA}{{\cal A}}
\newcommand{\cB}{{\cal B}}
\newcommand{\cX}{{\cal X}}

\newcommand{\cD}{{\cal D}}
\newcommand{\cY}{{\cal Y}}
\newcommand{\cT}{{\cal T}}
\newcommand{\cZ}{{\cal Z}}
\renewcommand{\S}{{\mathbb S}}

\newcommand{\Rep}{\mathrm{Rep}}
\newcommand{\Ind}{\mathop{\mathrm{Ind}}}
\newcommand{\Res}{\mathop{\mathrm{Res}}}
\renewcommand{\O}{{\mathrm O}}
\newcommand{\R}{{\mathbb R}}

\newcommand{\I}{{\mathbb I}}

\newcommand{\refthm}[1]{Theorem~\ref{thm:#1}}
\newcommand{\reflem}[1]{Lemma~\ref{lem:#1}}

\newcommand{\refcor}[1]{Corollary~\ref{cor:#1}}
\newcommand{\refclm}[1]{Claim~\ref{clm:#1}}

\newcommand{\refsec}[1]{Section~\ref{sec:#1}}
\newcommand{\refapp}[1]{Appendix~\ref{app:#1}}

\newcommand{\refeqn}[1]{(\ref{eqn:#1})}

\newcommand\qqAnd{\qquad\text{and}\qquad}

\def \norm|#1|{\left\| #1\right\|}

\usepackage{slashbox} 

\usepackage{tikz} 
\usetikzlibrary{shapes.geometric}
\usetikzlibrary{shapes.misc}
\usetikzlibrary{decorations.pathmorphing}
\usetikzlibrary{decorations.pathreplacing}
\usetikzlibrary{positioning,backgrounds,shadows}
\usetikzlibrary{fadings}

\usetikzlibrary{calc, trees, positioning, arrows, shapes, shapes.multipart, shadows, matrix, decorations.pathreplacing, decorations.pathmorphing}

\usetikzlibrary{shapes.geometric}
\usetikzlibrary{shapes.misc}
\usetikzlibrary{decorations.pathmorphing}
\usetikzlibrary{decorations.pathreplacing}
\usetikzlibrary{positioning,backgrounds,shadows}
\usetikzlibrary{decorations.markings} 

\begin{document}

\title{Quantum Adversary Lower Bound for Element Distinctness with Small Range}
\author{
\normalsize Ansis Rosmanis\\ [.5ex]
\small David R.\ Cheriton School of Computer Science and \\ 
\small Institute for Quantum Computing, University of Waterloo \\
\small \texttt{ansis.rosmanis@gmail.com}
}
\date{}

\maketitle

\begin{abstract}
The {\sc Element Distinctness} problem is to decide whether each character of an input string is unique. The quantum query complexity of {\sc Element Distinctness} is known to be $\Theta(N^{2/3})$; the polynomial method gives a tight lower bound for any input alphabet, while a tight adversary construction was only known for alphabets of size $\Omega(N^2)$.

We construct a tight $\Omega(N^{2/3})$ adversary lower bound for {\sc Element Distinctness} with minimal non-trivial alphabet size, which equals the length of the input. This result may help to improve lower bounds for other related query problems.
\end{abstract}

\section{Introduction and motivation}

\paragraph{Background.}

In quantum computation, one of the main questions that we are interested in is: What is the quantum circuit complexity of a given computational problem? This question is hard to answer, and so we consider an alternative question: What is the quantum query complexity of the problem?
For many problems, it is seemingly easier to (upper and lower) bound the number of times an algorithm requires to access the input rather than to bound the number of elementary quantum operations required by the algorithm. 
Nonetheless, the study of the quantum query complexity can give us great insights for the quantum circuit complexity. For example, a query-efficient algorithm for {\sc Simon's Problem} \cite{simon:powerOfQuantum} helped Shor to develop a time-efficient algorithm for factoring \cite{shor:factorization}.
On the other hand, $\tilde\Omega(N^{1/5})$ and $\Omega(N^{1/2})$ lower bounds on the (bounded error) quantum query complexity of the {\sc Set Equality} \cite{midrijanis:setEqual} and the {\sc Index Erasure} \cite{ambainis:symmetryAssisted} problems, respectively, ruled out certain approaches for constructing time-efficient quantum algorithms for
the {\sc Graph Isomorphism} problem.

Currently, two main techniques for proving lower bounds on quantum query complexity are the {\em polynomial method} developed by Beals, Buhrman, Cleve, Mosca, and de Wolf~\cite{beals:pol}, and the {\em adversary method} originally developed by Ambainis~\cite{ambainis:adv} in what later became known as the positive adversary method.
 The adversary method was later strengthened by H{\o}yer, Lee,
and \v Spalek~\cite{hoyer:advNegative} by allowing negative weights in the adversary matrix.
In recent results~\cite{reichardt:advTight, lee:stateConversion}, Lee, Mittal, Reichardt, \v Spalek, and Szegedy showed that, unlike the polynomial method \cite{ambainis:polVsQCC}, the general (i.e.,  strengthened) adversary method can give tight lower bounds for all problems. This is a strong incentive for the study of the adversary method.

\paragraph{Element Distinctness and Collision.}

Even though we know that tight adversary (lower) bounds exist for all query problems, for multiple problems we still do not know how to even construct adversary bounds that would match lower bounds obtained by other methods. For about a decade, {\sc Element Distinctness} and {\sc Collision} were prime examples of such problems. 
Given an input string $z\in \Sigma^N$, the {\sc Element Distinctness} problem is to decide whether each character of $z$ is unique, and the {\sc Collision} problem is its special case given a promise that each character of $z$ is either unique or appears in $z$ exactly twice. 
As one can think of $z$ as a function that maps $\{1,2,\ldots,N\}$ to $\Sigma$, the alphabet $\Sigma$ is often also called the {\em range}.

The quantum query complexity of these two problems is known.
Brassard, H{\o}yer, and Tapp first gave an $\O(N^{1/3})$ quantum query algorithm for {\sc Collision} \cite{brassard:collision}. Aaronson and Shi then gave a matching $\Omega(N^{1/3})$ lower bound for {\sc Collision} via the polynomial method, requiring that $|\Sigma|\geq 3N/2$ \cite{shi:collisionLower}. Due to a particular reduction from {\sc Collision} to {\sc Element Distinctness}, their lower bound also implied an $\Omega(N^{2/3})$ lower bound for {\sc Element Distinctness}, requiring that $|\Sigma|\in\Omega(N^2)$. Subsequently, Kutin (for {\sc{}Collision}) and Ambainis (for both) removed these requirements on the alphabet size 
 \cite{kutin:collisionLower,ambainis:collisionLower}.
Finally, Ambainis gave an $\O(N^{2/3})$ quantum query algorithm for {\sc Element Distinctness} based on a quantum walk \cite{ambainis:distinctness}, thus improving the best previously known $\O(N^{3/4})$ upper bound \cite{buhrman:distinctness}.

Hence, the proof of the $\Omega(N^{2/3})$ lower bound for {\sc Element Distinctness} with minimal non-trivial alphabet size $N$ (and, thus, any alphabet size) consists of three steps: 
 an $\Omega(N^{1/3})$ lower bound for {\sc Collision},
 a reduction from an $\Omega(N^{1/3})$ lower bound for {\sc Collision} to  an $\Omega(N^{2/3})$ lower bound for {\sc Element Distinctness} with the alphabet size $\Omega(N^2)$,
 and a reduction of the alphabet size. 
In this paper we prove the same result directly by providing an $\Omega(N^{2/3})$ general adversary bound for {\sc Element Distinctness} with the alphabet size $N$.

The problems of {\sc Set Equality}, {\sc $k$-Distinctness}, and {\sc $k$-Sum} are closely related to {\sc Collision} and {\sc Element Distinctness}. {\sc Set Equality} is a special case of {\sc Collision} given an extra promise that each character of the first half (and, thus, the second half) of the input string is unique. Given a constant $k$, the {\sc $k$-Distinctness} problem is to decide whether the input string contains some character at least $k$ times. 
 For {\sc $k$-Sum}, we assume that $\Sigma$ is an additive group and the problem is to decide if there exist $k$ numbers among $N$ that sum up to a prescribed number.

\paragraph{Recent adversary bounds.}

Due to the {\em certificate complexity barrier} \cite{zhang:advPower,spalek:advEquivalent},
 the positive weight adversary method fails to give a better lower bound for {\sc Element Distinctness} than $\Omega(N^{1/2})$. And similarly, due to the {\em property testing barrier} \cite{hoyer:advNegative}, it fails to give a better lower bound for {\sc Collision} than the trivial $\Omega(1)$.
Recently, Belovs gave an $\Omega(N^{2/3})$ general adversary bound for {\sc Element Distinctness} with a large $\Omega(N^2)$ alphabet size \cite{belovs:adv-el-dist}. In a series of works that followed, tight general adversary bounds were given for the {\sc $k$-Sum} \cite{spalek:kSumLower}, {\sc Certificate-Sum} \cite{belovs:nonAdaptiveLG}, and {\sc Collision} and {\sc Set Equality} problems \cite{rosmanis:collision}, all of them requiring that the alphabet size is large.
$\Omega(N^{k/(k+1)})$ and $\Omega(N^{1/3})$ lower bounds for {\sc $k$-Sum} and {\sc Set Equality}, respectively, were improvements over the best previously known lower bounds. (The $\Omega(N^{1/3})$ lower bound for {\sc Set Equality} was also independently proven by Zhandry \cite{zhandry:set-equal}; he used a completely different method, which did not require any assumptions on the alphabet size.)

The adversary lower bound for a problem is given via the {\em adversary matrix} (\refsec{adv}). The construction of the adversary matrix in all these recent (general) adversary bounds mentioned 
 has one idea in common: the adversary matrix is extracted from a larger matrix that has been constructed using, essentially, the Hamming association scheme \cite{godsil:assoc1}. The fact that we initially embed the adversary matrix in this larger matrix is the reason behind the requirement of the large alphabet size. More precisely, due to the birthday paradox, these adversary bounds require the alphabet $\Sigma$ to be large enough so that a randomly chosen string in $\Sigma^N$ with constant probability is a negative input of the problem.

Also, for these problems, all the negative inputs are essentially equally hard. 
 However, for {\sc $k$-Distinctness}, for example, the hardest negative inputs seem to be the ones in which each character appears $k-1$ times, and a randomly chosen negative input for {\sc $k$-Distinctness} is such only with a minuscule probability.
This might be a reason why an $\Omega(N^{2/3})$ adversary bound for {\sc $k$-Distinctness} \cite{spalek:adv-array} based on the idea of the embedding does not narrow the gap to the best known upper bound, $\O(N^{1-2^{k-2}/(2^k-1)})$ \cite{belovs:learningKDist}.
(The $\Omega(N^{2/3})$ lower bound was already known previously via the reduction from {\sc Element Distinctness} attributed to Aaronson in \cite{ambainis:distinctness}.)

\paragraph{Motivation for our work.}

In this paper we construct an explicit adversary matrix for {\sc Element Distinctness} with the alphabet size $|\Sigma|=N$ (and, thus, any alphabet size) yielding the tight $\Omega(N^{2/3})$ lower bound. We also provide certain ``tight'' conditions that every optimal adversary matrix for {\sc Element Distinctness} must satisfy,\footnote{%
Assuming, without loss of generality, that the adversary matrix has the symmetry given by the automorphism principle.
}
 therefore suggesting that every optimal adversary matrix for {\sc Element Distinctness} might have to be, in some sense, close to the adversary matrix that we have constructed.

The tight $\Omega(N^{k/(k+1)})$ adversary bound for {\sc $k$-Sum} by Belovs and \v{S}palek \cite{spalek:kSumLower} is an extension of Belovs' $\Omega(N^{2/3})$ adversary bound for {\sc Element Distinctness} \cite{belovs:adv-el-dist}, and it requires $|\Sigma|\in\Omega(N^k)$. We construct the adversary matrix for {\sc Element Distinctness} directly, without the embedding, therefore we do not require the condition $|\Sigma|\in\Omega(N^2)$ as in Belovs' adversary bound. 
We hope that this might help to reduce the required alphabet size in the $\Omega(N^{k/(k+1)})$ lower bound for {\sc $k$-Sum}.

As we mentioned before, an adversary matrix for {\sc $k$-Distinctness} based on the idea of the embedding might not be able to give tight lower bounds. On the other hand, in our construction we only assume that the adversary matrix is invariant under all index and all alphabet permutations, and that is something we can always do without loss of generality due to the {\em automorphism principle} \cite{hoyer:advNegative}---for {\sc Element Distinctness}, {\sc $k$-Distinctness}, and many other problems. 
 Hence, due to the optimality of the general adversary method, we know that one can construct a tight adversary bound for {\sc $k$-Distinctness} that satisfies these symmetries, and we hope
that our construction for {\sc Element Distinctness} might give insights in how to do that.

\paragraph{Structure of the paper.}

This paper is structured as follows. In \refsec{pre} we present the preliminaries of our work, including the adversary method, the automorphism principle, and the basics of the representation theory of the symmetric group. In \refsec{GammaBlocks} we show that the adversary matrix $\Gamma$ can be expressed as a linear combination of specific matrices. In this section we also present \refclm{GammaTransp}, which states what conditions every optimal adversary matrix for {\sc Element Distinctness} must satisfy; we prove this claim in the appendix. In \refsec{GammaViaGamma12} we show how to specify the adversary matrix $\Gamma$ via it submatrix $\Gamma_{1,2}$, which will make the analysis of the adversary matrix simpler. In \refsec{tools} we present tools for estimating the spectral norm of the matrix entrywise product of $\Gamma$ and the {\em difference matrix} $\Delta_i$, a quantity that is essential to the adversary method. In \refsec{sufficient} we use the conditions given by \refclm{GammaTransp} to construct an adversary matrix for {\sc Element Distinctness} with the alphabet size $N$, and we show that this matrix indeed yields the desired $\Omega(N^{2/3})$ lower bound. We conclude in \refsec{open} with open problems.

\newcommand\parA{\lambda}
\newcommand\parB{\mu}
\newcommand\parC{\nu}
\newcommand\parD{\zeta}
\newcommand\parE{\eta}
\newcommand\parF{\theta}

\newcommand\parDb{{\bar\parD}}
\newcommand\parEb{{\bar\parE}}
\newcommand\parFb{{\bar\parF}}

\newcommand\parDB{{\bar\parD}}
\newcommand\parEB{{\bar\parE}}
\newcommand\parFB{{\bar\parF}}

\newcommand\PiA{\Pi^\parA}
\newcommand\PiEb{\Pi^\parEb}

\section{Preliminaries} \label{sec:pre}

\subsection{Element distinctness problem}

Let $N$ be the length of the input and let $\Sigma$ be the input alphabet. Let $[i,N]=\{i,i+1,\ldots,N\}$ and $[N]=[1,N]$ for short.
Given an input string $z\in \Sigma^N$, the {\sc Element Distinctness} problem is to decide whether $z$ contains a collision or not, namely, weather there exist $i,j\in[N]$ such that $i\neq j$ and $z_i=z_j$. We only consider a special case of the problem where we are given a promise that the input contains at most one collision. This promise does not change the complexity of the problem \cite{ambainis:distinctness}. 

Let $\cD_1$ and $\cD_0$ be the sets of positive and negative inputs, respectively, that is, inputs with a unique collision and inputs without a collision. If $|\Sigma|<N$, then $\cD_0=\emptyset$, and the problem becomes trivial, therefore we consider the case when $|\Sigma|=N$. We have
\[
|\cD_1|=\binom{N}{2}\frac{|\Sigma|!}{(|\Sigma|-N+1)!} = \binom{N}{2}N! 
 \qquad\text{and}\qquad
 |\cD_0|=\frac{|\Sigma|!}{(|\Sigma|-N)!}=N!.
\] 

\subsection{Adversary method} \label{sec:adv}

The general adversary method gives optimal bounds for any quantum query problem.
Here we only consider the {\sc Element Distinctness} problem, so it suffices to define the adversary
 method for decision problems.
Let us think of a decision problem $p$ as a
 Boolean-valued function $p:\cD\rightarrow \{0,1\}$ with domain $\cD\subseteq\Sigma^N$, 
and let $\cD_1=p^{-1}(1)$ and $\cD_0=p^{-1}(0)$.

 An adversary matrix for a decision problem $p$ is a 
 real
 $|\cD_1|\times|\cD_0|$ matrix $\Gamma$ whose rows are labeled by the positive inputs
 $\cD_1$ and columns by the negative inputs $\cD_0$.
Let $\Gamma[\![x,y]\!]$ denote the entry of $\Gamma$ corresponding to the pair of inputs $(x,y)\in \cD_1\times \cD_0$.
 For $i\in[N]$, the difference matrices $\Delta_i$ and $\overline\Delta_i$ are the matrices of the same
 dimensions and the same row and column labeling as $\Gamma$ that are defined by
\[
 \Delta_i[\![x,y]\!] = \begin{cases}0,&\text{if }x_i=y_i,\\1,&\text{if }x_i\neq y_i,\end{cases}
\qquad\text{and}\qquad
 \overline\Delta_i[\![x,y]\!] = \begin{cases}1,&\text{if }x_i=y_i,\\0,&\text{if }x_i\neq y_i.\end{cases}
\]

\begin{thm}[Adversary bound~\cite{hoyer:advNegative,lee:stateConversion}]
\label{thm:adversary}
The quantum query complexity of the decision problem $p$ is $\Theta(\mathrm{Adv}(p))$, where $\mathrm{Adv}(p)$ is the optimal value
of the semi-definite program
\begin{equation}
\label{eqn:adv}
\begin{split}
 &{\mbox{\rm maximize }} \hspace{12.4pt} \norm|\Gamma|  \\ 
 &{\mbox{\rm subject to }}\quad \norm|\Delta_i\circ\Gamma|\leq 1\text{ for all }i\in[N],
\end{split}
\end{equation}
where the maximization is over all adversary matrices $\Gamma$ for $p$, $\norm|\cdot|$ is the spectral
 norm (i.e., the largest singular value), and $\circ$ is the entrywise matrix product.
\end{thm}

Every feasible solution to the semi-definite
 program \refeqn{adv} yields a lower bound on the quantum query complexity of $p$.
Note that we can choose any adversary matrix $\Gamma$ and scale it so that the condition
 $\norm|\Delta_i\circ\Gamma|\leq 1$ holds. In practice, we use the condition
 $\norm|\Delta_i\circ\Gamma|\in \mathrm{O}(1)$ instead of $\norm|\Delta_i\circ\Gamma|\leq 1$. 
Also note that $\Delta_i\circ\Gamma = \Gamma - \overline\Delta_i\circ\Gamma$.

\subsection{Symmetries of the adversary matrix}

It is known that we can restrict the maximization in Theorem \ref{thm:adversary} to adversary matrices
 $\Gamma$ satisfying certain symmetries.
Let $\S_A$ be the symmetric group of a finite set $A$, that is, the group whose elements are all the
 permutations of elements of $A$ and whose group operation is the composition of permutations.
 The automorphism principle \cite{hoyer:advNegative} implies that, without loss of 
 generality, we can assume that $\Gamma$ for {\sc Element Distinctness} is fixed under all index and all alphabet permutations.
Namely, index permutations $\pi\in\S_{[N]}$ and alphabet permutations $\tau\in\S_\Sigma$ act on input strings $z\in\Sigma^N$ in the natural way:
\begin{alignat*}{3}
&\pi\in\S_{[N]} &\;\;:\;\;&z=(z_1,\ldots,z_N) \;\;\mapsto\;\;
z_\pi=\big(z_{\pi^{-1}(1)},\ldots,z_{\pi^{-1}(N)}\big),\\
&\tau\in\S_\Sigma&\;\;:\;\; &z=(z_1,\ldots,z_N) \;\; \mapsto \;\;
z^\tau =\big(\tau(z_1),\ldots,\tau(z_N)\big).
\end{alignat*}
The actions of $\pi$ and $\tau$ commute: we have $(z_\pi)^\tau=(z^\tau)_\pi$, which we denote by $z^\tau_\pi$ for short.
The automorphism principle implies that we can assume
\begin{equation}\label{eqn:auto1}
\Gamma[\![x,y]\!] = \Gamma[\![x^\tau_\pi,y^\tau_\pi]\!]
\end{equation}
for all $x\in \cD_1$, $y\in \cD_0$, $\pi\in \S_{[N]}$, and $\tau\in \S_\Sigma$.

Let $\cX\cong\R^{|\cD_1|}$ and $\cY\cong\R^{|\cD_0|}$ be the vector spaces
 corresponding to the positive and the negative inputs, respectively.
(We can view $\Gamma$ as a linear operator that maps $\cY$ to $\cX$.) 
Let $U_\pi^\tau$ and $V_\pi^\tau$ be the permutation matrices that respectively act on the spaces $\cX$ and $\cY$ and that map every $x\in \cD_1$ to
$x_\pi^\tau$ and every $y\in \cD_0$ to $y_\pi^\tau$.
Then (\ref{eqn:auto1}) is equivalent to 
\begin{equation}\label{eqn:auto2}
   U_\pi^\tau\Gamma=\Gamma V_\pi^\tau
\end{equation}
for all $\pi\in \S_{[N]}$, and $\tau\in \S_\Sigma$.
Both $U$ and $V$ are representations of $\S_{[N]}\times\S_\Sigma$.

\subsection{Representation theory of the symmetric group} \label{sec:rep}

Let us present the basics of the representation theory of the symmetric group. 
(For a detailed study of the representation theory of the symmetric group, refer to~\cite{james:symmetric,sagan:symmetric};
for the fundamentals of the representation theory of finite groups, refer to~\cite{serre:representation}.)

Up to isomorphism, there is one-to-one correspondence between the {\em irreps} (i.e., irreducible representations) of $\S_A$ and $|A|$-box {\em Young diagrams}, and we often use these two terms interchangeably.
We use
 $\parD$, $\parE$, and $\parF$ to denote Young diagrams having $o(N)$ boxes, $\parA$, $\parB$, and $\parC$ to denote Young diagrams having $N$, $N-1$, and $N-2$ boxes, respectively, and $\rho$ and $\sigma$ for general statements and other purposes.

 Let $\rho\vdash M$ denote that $\rho$ is an $M$-box Young diagram.
For a Young diagram $\rho$, 
 let $\rho(i)$ and $\rho\!^\top\!(j)$ denote the number of boxes in the $i$-th row and  $j$-th column of $\rho$, respectively. 
We write $\rho=(\rho(1),\rho(2),\ldots,\rho(r))$, where $r=\rho\!^\top\!(1)$ is the number of rows in $\rho$, and, given $M\ge \rho(1)$, let $(M,\rho)$ be short for $(M,\rho(1),\rho(2),\ldots,\rho(r))$.

We say that a box $(i,j)$ is present in $\rho$ and write $(i,j)\in \rho$ if $\rho(i)\ge j$ (equivalently, $\rho\!^\top\!(j)\ge i$). 
The {\em hook-length} $h_\rho(b)$ of a box $b$ is the sum of the number of boxes on the right from $b$ in the same row (i.e., $\rho(i)-j$) and the number of boxes below $b$ in the same column (i.e., $\rho\!^\top\!(j)-i$) plus one (i.e., the box $b$ itself). 
 The dimension of the irrep corresponding to $\rho$ is given by the {\em hook-length formula}:
\begin{equation}
\label{eqn:hook}
\dim\rho=|\rho|!\big/h(\rho),
\qquad\text{where}\qquad
h(\rho)= \prod\nolimits_{(i,j)\in\rho} h_\rho(i,j)
\end{equation}
and $|\rho|$ is the number of boxes in $\rho$.

Let $\sigma<\rho$ and $\sigma\ll\rho$ denote that a Young diagram $\sigma$ is obtained from $\rho$ by removing exactly one box and exactly two boxes, respectively. 
Given $\sigma\ll\rho$, let us write $\sigma\ll_r\rho$ or $\sigma\ll_c\rho$ if the two boxes removed from $\rho$ to obtain $\sigma$ are, respectively, in different rows or different columns.
Let $\sigma\ll_{rc}\rho$ be short for $(\sigma\ll_r\rho)\&(\sigma\ll_c\rho)$.
The {\em distance} 
 between two boxes $b=(i,j)$ and $b'=(i',j')$ is defined as
\(
|i'-i|+|j-j'|
\).
Given $\sigma\ll_{rc}\rho$, let $d_{\rho,\sigma}\geq 2$ be the distance between the two boxes that we remove from $\rho$ to obtain $\sigma$.

 The {\em branching rule} states that the restriction of an irrep $\rho$ of $\S_{A}$ to $\S_{A\setminus\{a\}}$, where $a\in A$, is
\[
\Res\nolimits_{\,\S_{A\setminus\{a\}}}^{\,\S_A} \!\rho \cong \bigoplus\nolimits_{\sigma<\rho}\sigma.
\]
The more general {\em Littlewood--Richardson rule} implies that the restriction of an irrep $\rho$ of $\S_{A}$ to $\S_{\{a,b\}}\times\S_{A\setminus\{a,b\}}$, where $a,b\in A$,
is
\[
\Res\nolimits_{\,\S_{\{a,b\}}\times\S_{A\setminus\{a,b\}}}^{\,\S_A} \!\rho \cong
\bigoplus\nolimits_{\sigma\ll_c\rho}(id\times\sigma)
\oplus 
\bigoplus\nolimits_{\sigma'\ll_r\rho}(sgn\times\sigma'),
\]
where $id=(2)$ and $sgn=(1,1)$ are the trivial and the sign representation of $\S_{\{a,b\}}$, respectively.
{\em{}Frobenius reciprocity} then tells us that the ``opposite'' happens when we induce an irrep of
 $\S_{A\setminus\{a\}}$ or $\S_{\{a,b\}}\times\S_{A\setminus\{a,b\}}$ to $\S_A$.
 
Given $\ell\in\{0,1,2,3\}$, a set $A=[N]$ or $A=\Sigma$, its subset $A\setminus\{a_1,\ldots,a_\ell\}$, and $\rho\vdash N-\ell$, let us write $\rho_{a_1\!\ldots a_\ell}$ if we want to stress that we think of $\rho$ as an irrep of $\S_{A\setminus\{a_1,\ldots,a_\ell\}}$. We omit the subscript if $\ell=0$ or when $\{a_1,\ldots,a_\ell\}$ is clear from the context. To lighten the notations, given $k\in o(N)$ and $\parE\vdash k$, let $\parEb_{a_1\!\ldots a_\ell}=(N-\ell-k,\parE)_{a_1\!\ldots a_\ell}\vdash N-\ell$;
here we omit the subscript if and only if $\ell=0$.

\subsection{Transporters}

Suppose we are given a group $G$, and let $\xi_1$ and $\xi_2$ be two isomorphic irreps of $G$ acting on spaces $\cZ_1$ and $\cZ_2$, respectively. 
Up to a global phase (i.e., a scalar of absolute value $1$), there exists a unique isomorphism $T_{2\leftarrow 1}$ from $\xi_1$ to $\xi_2$ that satisfies $\|T_{2\leftarrow 1}\|=1$. We call this 
isomorphism a {\em transporter} from $\xi_1$ to $\xi_2$ (or, from $\cZ_1$ to $\cZ_2$).

In this paper we only consider unitary representations and real vector spaces, therefore all singular values of $T_{2\leftarrow 1}$ are equal to $1$ and, for the global phase, we have to choose only between $\pm 1$. We always choose the global phases so that they respect inversion and composition, namely, so that $T_{1\leftarrow 2}T_{2\leftarrow 1}$ is the identity matrix on $\cZ_1$ and $T_{3\leftarrow 2}T_{2\leftarrow 1}=T_{3\leftarrow 1}$, where $\xi_3$ is an irrep isomorphic to $\xi_1$ and $\xi_2$.

\section{\texorpdfstring{Building blocks of $\Gamma$}{Building blocks of the adversary matrix}} \label{sec:GammaBlocks}

\subsection{Decomposition of $U$ and $V$ into irreps}

Without loss of generality, let us assume that the adversary matrix $\Gamma$ for the {\sc Element Distinctness} problem satisfy the symmetry \refeqn{auto2} given by the automorphism principle.
Both $U$ and $V$ are representations of $\S_{[N]}\times \S_\Sigma$ and, due to Schur's lemma, we want to see what irreps of $\S_{[N]}\times \S_\Sigma$ occur in both $U$ and $V$.
It is also convenient to consider $U$ and $V$ as representations of just $\S_{[N]}$ or just $\S_\Sigma$. 
\begin{clm} \label{clm:VDecomp}
$V$ decomposes into irreps of $\S_{[N]}\times \S_\Sigma$ as
\(
V\cong\bigoplus\nolimits_{\parA\vdash N}\parA\times\parA.
\)
\end{clm}

\begin{proof}
As a representation of $\S_{[N]}$ and $\S_\Sigma$, respectively, $V$ is isomorphic to the regular representation of $\S_{[N]}$ and $\S_\Sigma$. For every $y\in\cD_0$ and every $\pi\in\S_{[N]}$, there is a unique $\tau\in\S_\Sigma$ such that $y_\pi=y^\tau$, and $\pi$ and $\tau$ belong to isomorphic conjugacy classes. Thus, for every $\parA$, the {\em isotypical subspace} of $\cY$ corresponding to $\parA$ (i.e., the subspace corresponding to all irreps isomorphic to $\parA$) is the same for both $\S_{[N]}$ and $\S_\Sigma$ \cite[Section 2.6]{serre:representation}.
Since $V$ is isomorphic to the regular representation, the dimension of this subspace is $(\dim\parA)^2$, which is exactly the dimension of the irrep $\parA\times\parA$ of $\S_{[N]}\times \S_\Sigma$.
\end{proof}

Now let us address $U$, which acts on the space $\cX$ corresponding to the positive inputs $x\in\cD_1$.
Let us decompose $\cD_1$ as a disjoint union of $\binom{N}{2}$ sets $\cD_{i,j}$, where $\{i,j\}\subset[N]$ and $\cD_{i,j}$ is the set of all $x\in\cD_1$ such that $x_i=x_j$.
 Let us further decompose $\cD_{i,j}$ as a disjoint union of $\binom{N}{2}$ sets $\cD_{i,j}^{s,t}$, where $\{s,t\}\subset\Sigma$ and $\cD_{i,j}^{s,t}$ is the set of all $x\in\cD_{i,j}$ that does not contain $s$ and contains $t$ twice or vice versa. Let $\cX_{i,j}$ and $\cX_{i,j}^{s,t}$ be the subspaces of $\cX$ that correspond to the sets $\cD_{i,j}$ and $\cD_{i,j}^{s,t}$, respectively.
The space $\cX_{i,j}^{s,t}$
is invariant under the action of the subgroup $\S_{i,j}^{s,t}$ defined as
\[
\S_{i,j}^{s,t}=
(\S_{\{i,j\}}\times\S_{[N]\setminus\{i,j\}})
\times
(\S_{\{s,t\}}\times\S_{\Sigma\setminus\{s,t\}}),
\]
namely, $U_\pi^\tau\cX_{i,j}^{s,t}=\cX_{i,j}^{s,t}$ for all $(\pi,\tau)\in \S_{i,j}^{s,t}$.
Therefore $U$ restricted to the subspace $\cX_{i,j}^{s,t}$ is a representation of $\S_{i,j}^{s,t}$, and, similarly to Claim \ref{clm:VDecomp}, it decomposes into irreps as
\begin{equation} \label{eqn:SijstIrrep}
\bigoplus\nolimits_{\parC\vdash N-2} \big(id\times\parC\big)\times\big((id\oplus sgn)\times\parC\big).
\end{equation}
To see how $U$ decomposes into irreps of $\S_{[N]}\times\S_\Sigma$, we induce  the representation (\ref{eqn:SijstIrrep}) from $\S_{i,j}^{s,t}$ to $\S_{[N]}\times\S_\Sigma$.

The Littlewood--Richardson rule implies that an irrep of $\S_{[N]}\times\S_\Sigma$ isomorphic to $\parA\times\parA$ 
can occur in $U$ due to one of the following scenarios.
\begin{itemize}
\item If $\parC\ll_c\parA$ and $\parC\not\!\ll_r\parA$  (i.e., $\parC$ is obtained from $\parA$ by removing two boxes in the same row), then $\parA\times\parA$ occurs once in the induction of $(id\times\parC)\times(id\times\parC)$. Let $\cX^\parA_{id,\parC}$ denote the subspace of $\cX$ corresponding to this instance of $\parA\times\parA$.
\item If $\parC\ll_{rc}\parA$, then $\parA\times\parA$ occurs once in the induction of $(id\times\parC)\times(id\times\parC)$ and once in the induction of $(id\times\parC)\times(sgn\times\parC)$. Let $\cX^\parA_{id,\parC}$ and $\cX^\parA_{sgn,\parC}$ denote the respective subspaces of $\cX$ corresponding to these instances of $\parA\times\parA$.
\end{itemize}
Note: the subspaces $\cX^\parA_{id,\parC}$ and $\cX^\parA_{sgn,\parC}$ are independent from the choice of $\{i,j\}\subset[N]$ and $\{s,t\}\subset\Sigma$.

\subsection{\texorpdfstring{$\Gamma$ as a linear combination of transporters}{Adversary matrix $\Gamma$ as a linear combination of transporters}}

Let $\Xi^\parA_{id,\parC}$ and $\Xi^\parA_{sgn,\parC}$ denote the transporters from the unique instance of $\parA\times\parA$ in $\cY$ to the subspaces $\cX^\parA_{id,\parC}$ and $\cX^\parA_{sgn,\parC}$, respectively. We will specify the global phases of these transporters in \refsec{Gamma12ViaProj}.
 We consider $\Xi^\parA_{id,\parC}$ and $\Xi^\parA_{sgn,\parC}$ as matrices of dimensions
 $\binom{N}{2}N!\times N!$ and rank $(\dim\parA)^2$.
Schur's lemma implies that, due to (\ref{eqn:auto2}), we can express $\Gamma$ as a linear combination of these transporters. Namely,
\begin{equation}\label{eqn:GammaBlocks}
\Gamma = \sum_{\parA\vdash N}
\Big(
\sum_{\parC\ll_c \parA}     \beta^\parA_{id,\parC}\Xi^\parA_{id,\parC} + 
\sum_{\parC\ll_{rc} \parA}\beta^\parA_{sgn,\parC}\Xi^\parA_{sgn,\parC}
\Big),
\end{equation}
where the coefficients $\beta^\parA_{id,\parC}$ and $\beta^\parA_{sgn,\parC}$ are real.

\newcommand{\ch}{{\checkmark}}
\newcommand{\chch}{{\checkmark\!\!\!\checkmark}}
\newcommand{\rch}{{\ch$\!_0$}}
\newcommand{\bchch}{{\chch$\!_1$}}
\newcommand{\gch}{{\ch$\!_2$}}
\newcommand{\gchch}{{\chch$\!_2$}}

Thus we have reduced the construction of the adversary matrix $\Gamma$ to choosing the coefficients $\beta$ of the transporters in (\ref{eqn:GammaBlocks}). 
To illustrate what are the available transporters, let us consider the last four $(N-2)$-box Young diagrams $\parC$ of the {\em lexicographical order}---$(N-2)$, $(N-3,1)$, $(N-4,2)$, and $(N-4,1,1)$---and all $\parA$ that are obtained from these $\parC$ by adding two boxes in different columns. Table \ref{table} shows pairs of $\parA$ and $\parC$ for which we have both $\Xi^\parA_{id,\parC}$ and $\Xi^\parA_{sgn,\parC}$ available for the construction of $\Gamma$ (double check mark ``\chch'') or just $\Xi^\parA_{id,\parC}$ available (single check mark ``\ch'').

\begin{table}[here]
\centering
\begin{tabular}{l|ccccc}
\hspace{23pt}\backslashbox{$\parA$}{$\parC$}\hspace{-10pt} & $(N\!-\!2)$ & $(N\!-\!3,1)$ & $(N\!-\!4,2)$ & $(N\!-\!4,1,1)$ 
\\
\hline
$(N)$ &  \rch\\
$(N\!-\!1,1)$ & \bchch & \rch \\
$(N\!-\!2,2)$ & \gch & \bchch & \rch \\
$(N\!-\!2,1,1)$ & & \bchch & & \rch \\
$(N\!-\!3,3)$ & & \gch & \bchch &  \\
$(N\!-\!3,2,1)$ & & \gchch & \bchch & \bchch \\
$(N\!-\!3,1,1,1)$ & & & & \bchch \\
$(N\!-\!4,4)$ & & &\gch &  \\
$(N\!-\!4,3,1)$ & & &\gchch & \gch \\
$(N\!-\!4,2,2)$ & & & \gch & \\
$(N\!-\!4,2,1,1)$ & & & & \gchch 
\end{tabular}
\caption{Available operators for the construction of $\Gamma$. We distinguish three cases: both $\parA$ and $\parC$ are the same below the first row (label ``\rch''), $\parA$ has one box more below the first row than $\parC$ (label ``\bchch''), $\parA$ has two boxes more below the first row than $\parC$ (labels ``\gch'' and ``\gchch'').}
\label{table}
\end{table}

Due to the symmetry, $\|\Delta_i\circ\Gamma\|$ is the same for all $i\in[N]$, so, from now on, let us only consider $\Delta_1\circ\Gamma$. We want to choose the coefficients $\beta$ so that $\|\Gamma\|\in\Omega(N^{2/3})$ and $\|\Delta_1\circ\Gamma\|\in\O(1)$. The automorphism principle also implies (see \cite{hoyer:advNegative}) that we can assume that the principal left and right singular vectors of $\Gamma$ are the all-ones vectors, which correspond to $\Xi^{(N)}_{id,(N-2)}$. We thus choose $\beta^{(N)}_{id,(N-2)}\in\Theta(N^{2/3})$. 

In order to understand how to choose the coefficients $\beta$, in \refapp{necessary} we prove the following claim, which relates all the coefficients of transporters of Table \ref{table} and more.

\begin{clm} \label{clm:GammaTransp}
Suppose $\Gamma$ is given as in (\ref{eqn:GammaBlocks}) and $\beta^{(N)}_{id,(N-2)}=N^{2/3}$. 
Consider $\parA\vdash N$ that has $\O(1)$ boxes below the first row and $\parC\ll_c\parA$.
In order for $\|\Delta_1\circ\Gamma\|\in\O(1)$ to hold, 
 we need to have 
\begin{enumerate}
\item $\beta^\parA_{id,\parC}=N^{2/3}+\O(1)$ if $\parA$ and $\parC$ are the same below the first row,
\item $\beta^\parA_{id,\parC},\beta^\parA_{sgn,\parC}=c^\parA_\parC N^{1/6}+\O(1)$ if $\parA$ has one box more below the first row than $\parC$,
where $c^\parA_\parC$ is a constant depending only on the part of $\parA$ and $\parC$ below the first row,\footnote{Let $\hat\parA$ and $\hat\parC$ be the part of $\parA$ and $\parC$ below the first row, respectively. Then $c^\parA_\parC=\sqrt{h(\hat\parA)/h(\hat\parC)}=\sqrt{N\dim\parC/\dim\parA}+\O(1/N)$.}
\item $\beta^\parA_{id,\parC},\beta^\parA_{sgn,\parC}=\O(1)$ if $\parA$ has two boxes more below the first row than $\parC$.
\end{enumerate}
\end{clm} 

Note that we always have the freedom of changing (a constant number of) coefficients $\beta$ 
up to an additive term of $\O(1)$ because of the fact that
\begin{equation} \label{eqn:gamma2}
\gamma_2(\Delta_1)=\max_B\big\{\,\|\Delta_1\circ B\|\,:\,\|B\|\leq1\,\big\} \leq 2
\end{equation}
(see \cite{hoyer:advNegative} for this and other facts about the $\gamma_2$ norm).
We will use this fact again in \refsec{sufficient}.

\section{\texorpdfstring{Specification of $\Gamma$ via $\Gamma_{1,2}$}{Specification of the adversary matrix via its restriction}} \label{sec:GammaViaGamma12}

Due to the symmetry (\ref{eqn:auto1}), it suffices to specify a single row of the adversary matrix $\Gamma$ in order to specify the whole matrix. For the convenience, let us instead specify $\Gamma$ via specifying its $(N!\times N!)$-dimensional submatrix $\Gamma_{1,2}$---for $\{i,j\}\subset[N]$, we define $\Gamma_{i,j}$ to be the submatrix of $\Gamma$ that corresponds to the rows labeled by $x\in\cD_{i,j}$, that is, positive inputs $x$ with $x_i=x_j$.
We think of $\Gamma_{i,j}$ both as an $N!\times N!$ square matrix and as a matrix of the same dimensions as $\Gamma$ that is obtained from $\Gamma$ by setting to zero all the $\big(\binom{N}{2}-1\big)N!$ rows that correspond to $x\notin\cD_{i,j}$.

\subsection{\texorpdfstring{Necessary and sufficient symmetries of $\Gamma_{1,2}$}{Necessary and sufficient symmetries of the restriction}}

For all
$(\pi,\tau)\in(\S_{\{1,2\}}\times\S_{[3,N]})\times\S_\Sigma$,
we have
$U_\pi^\tau\cX_{1,2}=\cX_{1,2}$ and, therefore,
$U_\pi^\tau\Gamma_{1,2}=\Gamma_{1,2} V_\pi^\tau$. 
This is the necessary and sufficient symmetry that $\Gamma_{1,2}$ must satisfy in order for $\Gamma$ to be fixed under all index and alphabet permutations.
Since $U_{(12)}\Gamma_{1,2}=\Gamma_{1,2}$, we also have $\Gamma_{1,2}V_{(12)}=\Gamma_{1,2}$. We have
\begin{equation}\label{eqn:GammaViaBlocks}
\Gamma
=\sum_{\{i,j\}\subset[N]}\Gamma_{i,j}
= \sum_{\pi\in R}
U_{\pi} \Gamma_{1,2} V_{\pi^{-1}}
=\binom{N}{2}\frac{1}{N!}\sum_{\pi\in\S_{[N]}}
U_{\pi} \Gamma_{1,2} V_{\pi^{-1}},
\end{equation}
where $R=\Rep(\S_{[N]}/(\S_{\{1,2\}}\times\S_{[3,N]}))$ is a transversal of the left cosets of $\S_{\{1,2\}}\times\S_{[3,N]}$ in $\S_{[N]}$.

 Let $f$ be a bijection between $\cD_0$ and $\cD_{1,2}$ defined as
\[
f : \cD_0\rightarrow \cD_{1,2} : (y_1,y_2,y_3,\ldots,y_N) \mapsto (y_1,y_1,y_3,\ldots,y_N),
\]
and let $F$ be the corresponding permutation matrix mapping $\cY$ to $\cX_{1,2}$.
Let us order rows and columns of $\Gamma_{1,2}$ so that they correspond to $f(y)$ and $y$, respectively, where we take $y\in\cD_0$ in the same order for both (see Figure \ref{fig:GBlock}). 
 Hence, $F$ becomes the identity matrix on $\cY$ (from this point onward, we essentially think of $\cX_{1,2}$ and $\cY$ as the same space). Let us denote this identity matrix by $\I$.

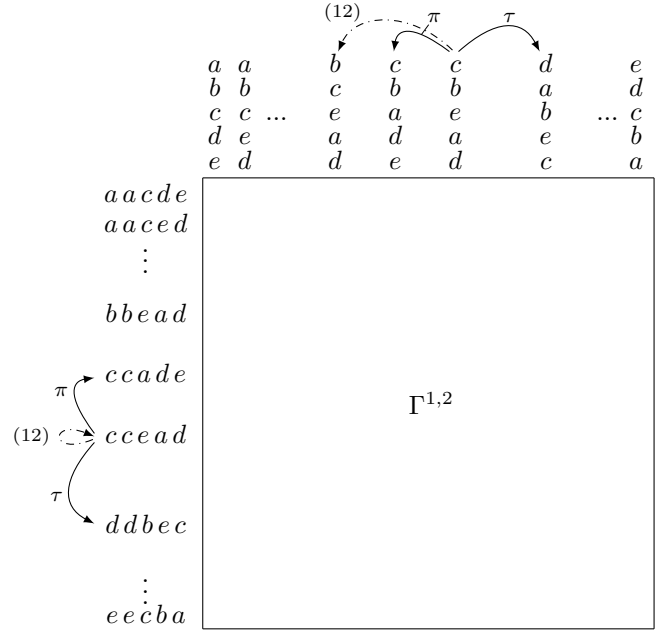
\begin{figure}[tbh]
\centering
\begin{tikzpicture}[xscale=0.4, yscale=0.4]
 \draw (1.9,0.5) -- (1.9,-14.5) ;
 \draw (1.9,0.5) -- (16.9,0.5);
 \draw (16.9,0.5) -- (16.9,-14.5) ;
 \draw (1.9,-14.5) -- (16.9,-14.5);
  \node at (9.5,-7) {$\Gamma^{1,2}$} ;
 
\node at (0,0) {$a\,a\,c\,d\,e$};
 \node at (0,-1) {$a\,a\,c\,e\,d$};
 \node at (0,-2) {$\vdots$};
 
 \node at (0,-4) {$b\,b\,e\,a\,d$};
 \node at (0,-6) {$c\,c\,a\,d\,e$};
\node at (0,-8) {$c\,c\,e\,a\,d$};
\node at (0,-11) {$d\,d\,b\,e\,c$};

\node at (0,-13) {$\vdots$};
\node at (0,-14) {$e\,e\,c\,b\,a$};

 \node[above] at (2.3,3.7) {$a$};
 \node[above] at (2.3,2.9) {$b$};
 \node[above] at (2.3,2.1) {$c$};
 \node[above] at (2.3,1.3) {$d$};
 \node[above] at (2.3,0.5) {$e$};

 \node[above] at (3.3,3.7) {$a$};
 \node[above] at (3.3,2.9) {$b$};
 \node[above] at (3.3,2.1) {$c$};
 \node[above] at (3.3,1.3) {$e$};
 \node[above] at (3.3,0.5) {$d$};

\node[above] at (4.35,2.1) {$...$};
\node[above] at (15.35,2.1) {$...$};

 \node[above] at (6.3,3.7) {$b$};
 \node[above] at (6.3,2.9) {$c$};
 \node[above] at (6.3,2.1) {$e$};
 \node[above] at (6.3,1.3) {$a$};
 \node[above] at (6.3,0.5) {$d$};

 \node[above] at (8.3,3.7) {$c$};
 \node[above] at (8.3,2.9) {$b$};
 \node[above] at (8.3,2.1) {$a$};
 \node[above] at (8.3,1.3) {$d$};
 \node[above] at (8.3,0.5) {$e$};

 \node[above] at (10.3,3.7) {$c$};
 \node[above] at (10.3,2.9) {$b$};
 \node[above] at (10.3,2.1) {$e$};
 \node[above] at (10.3,1.3) {$a$};
 \node[above] at (10.3,0.5) {$d$};

 \node[above] at (13.3,3.7) {$d$};
 \node[above] at (13.3,2.9) {$a$};
 \node[above] at (13.3,2.1) {$b$};
 \node[above] at (13.3,1.3) {$e$};
 \node[above] at (13.3,0.5) {$c$};

 \node[above] at (16.3,3.7) {$e$};
 \node[above] at (16.3,2.9) {$d$};
 \node[above] at (16.3,2.1) {$c$};
 \node[above] at (16.3,1.3) {$b$};
 \node[above] at (16.3,0.5) {$a$};


\node at (-2.8,-6.6) {\footnotesize{$\pi$}};
\node at (-3.8,-8.1) {\scriptsize{$(12)$}};
\node at (-2.95,-10.1) {\footnotesize{$\tau$}};

\node at (6.55,6.0) {\scriptsize{$(12)$}};
\node at (9.6,5.85) {\footnotesize{$\pi$}}; \draw [very thin] (9.16,5.256) -- +(0.25,0.34);
\node at (12.1,5.85) {\footnotesize{$\tau$}};

\begin{scope}[->,>=latex]
    \draw [dashdotted,->] (-1.75,-8.2) .. controls (-3.2,-8.8) and (-3.2,-7.5) .. (-1.7,-8.1);
    \draw [->] (-1.7,-8) .. controls (-2.7,-6.666) and (-2.4,-6.25) .. (-1.7,-6.15);
    \draw [->] (-1.7,-8.3) .. controls (-2.8,-9.533) and (-2.8,-10.433) .. (-1.7,-11);

    \draw [dashdotted,->] (10.2,4.75) .. controls (8.8,6.1) and (6.95,6) .. (6.4,4.7);
    \draw [] (10.1,4.7) .. controls (8.766,5.7) and (8.35,5.4) .. (8.25,4.7);
    \draw [->] (10.4,4.7) .. controls (11.633,5.8) and (12.533,5.8) .. (13.1,4.7);
\end{scope}
\end{tikzpicture}
\caption{%
Symmetries of $\Gamma_{1,2}$ for $N=5$ and $\Sigma=\{a,b,c,d,e\}$. 
With respect to the bijection $f$, the order of rows and columns matches. 
The solid arrows show that $U^\tau$ and $V^\tau$ act symmetrically on $\Gamma_{1,2}$ (here we use $\tau=(aeb)(cd)\in\S_\Sigma$),
and so do $U_\pi$ and $V_\pi$ for $\pi\in\S_{[3,N]}$ (here we use $\pi=(354)$).
However, as shown by the dash-dotted arrows, $U_{(12)}$ acts as the identity on the rows, while $V_{(12)}$ transposes the columns.}
\label{fig:GBlock}
\end{figure}

For all $(\pi,\tau)\in\S_{[3,N]}\times\S_\Sigma$ we have $f(y_\pi^\tau)=(f(y))_\pi^\tau$ and, thus, $V_\pi^\tau=F V_\pi^\tau=U_\pi^\tau F=U_\pi^\tau$, where we consider the restriction of $U_\pi^\tau$ to $\cX_{1,2}$. Note that $U_{(12)}=\I$ on $\cX_{1,2}$, while $V_{(12)}\neq\I$. Hence now the two necessary and sufficient symmetries that $\Gamma_{1,2}$ must satisfy are
\begin{equation}
\label{eqn:Gamma12Sym}
V_\pi^\tau \Gamma_{1,2} = \Gamma_{1,2} V_\pi^\tau
\quad
\text{for all}
\quad
(\pi,\tau)\in\S_{[3,N]}\times\S_\Sigma
\qquad\text{and}\qquad
\Gamma_{1,2}V_{(12)}=\Gamma_{1,2}.
\end{equation}
Figure \ref{fig:GBlock} illustrates these symmetries.

\subsection{Labeling of projectors and transporters}

We use $\Pi$, with some subscripts and superscripts, to denote operators acting on $\cY$; we use subscripts for irreps of index permutations and superscripts for irreps of alphabet permutations.  We also think of each such an operator $\Pi$ to map $\cY$ to $\cX_{1,2}$ and vice versa (technically, $F\Pi$ and $\Pi F^*$, respectively). 

Let $\Pi_{id}=(\I+V_{(12)})/2$ and $\Pi_{sgn}=(\I-V_{(12)})/2$ denote the projectors on the isotypical subspaces of $\cY$ corresponding to irreps $id=(2)$ and $sgn=(1,1)$ of $\S_{\{1,2\}}$, respectively. Let $\Pi_{\rho_{i_1\!\ldots i_\ell}}$ and $\Pi^{\sigma_{s_1\!\ldots s_m}}$ denote the projectors on the isotypical subspaces corresponding to an irrep 
$\rho$ of $\S_{[N]\setminus\{i_1,\ldots,i_\ell\}}$ and an irrep 
$\sigma$  of $\S_{\Sigma\setminus\{s_1,\ldots,s_m\}}$, respectively.
Note that $\Pi_{\rho_{i_1\!\ldots i_\ell}}$ and $\Pi^{\sigma_{s_1\!\ldots s_m}}$ commute, and let
\[
  \Pi_{\rho_{i_1\!\ldots i_\ell}}^{\sigma_{s_1\!\ldots s_m}} =
  \Pi_{\rho_{i_1\!\ldots i_\ell}} \Pi^{\sigma_{s_1\!\ldots s_m}} =
  \Pi^{\sigma_{s_1\!\ldots s_m}}\Pi_{\rho_{i_1\!\ldots i_\ell}},
\]
which is the projector on the isotypical subspace corresponding to the irrep $\rho\times\sigma$ of 
$\S_{[N]\setminus\{i_1,\ldots,i_\ell\}}\times\S_{\Sigma\setminus\{s_1,\ldots,s_m\}}$
 (note: this subspace may contain multiple instances of the irrep).
In general, when multiple such projectors 
 mutually commute, we denote their product with a single $\Pi$ whose subscript and superscript is, respectively, a concatenation of the subscripts and superscripts of these projectors. 
For example, $\Pi_{id,\parC_{12}}^\parA=\Pi_{id}^\parA\Pi_{\parC_{12}}\Pi^\parA$
(note: $\Pi^\parA$ corresponds to an irrep $\parA$ of $\S_{\Sigma\setminus\emptyset}=\S_\Sigma$).

Suppose that $\Pi^\parA_{sub}$ and $\Pi^\parA_{sub'}$ are two projectors each projecting onto a single instance of an irrep $\rho_{i_1\!\ldots i_\ell}\times \parA$ of $\S_{[N]\setminus\{i_1,\ldots,i_\ell\}}\times\S_\Sigma$, where $sub$ and $sub'$ are subscripts determining these instances.
 Then let $\Pi^\parA_{sub'\leftarrow sub}$ denote the transporter from the instance corresponding to $\Pi^\parA_{sub}$ to one corresponding to $\Pi^\parA_{sub'}$. Let
\(
\Pi^\parA_{sub'\leftrightarrow sub} = \Pi^\parA_{sub'\leftarrow sub} + \Pi^\parA_{sub\leftarrow sub'}
\)
for short.

\subsection{\texorpdfstring{Decomposition of $\Gamma_{1,2}$ into projectors and transporters}{Decomposition of the restriction into projectors and transporters}} \label{sec:Gamma12ViaProj}

Due to \refeqn{Gamma12Sym}, we can express $\Gamma_{1,2}$ as a linear combination of projectors onto irreps and transporters between isomorphic irreps of $\S_{[3,N]}\times\S_\Sigma$.
Due to \refeqn{Gamma12Sym} we also have $\Gamma_{1,2}\Pi_{id}=\Gamma_{1,2}$ and $\Gamma_{1,2}\Pi_{sgn}=0$. Claim \ref{clm:VDecomp} states that 
 $\I=\sum_{\parA\vdash N}\Pi_\parA^\parA$, and  
we have
$\Pi_\parA^\parA = \sum\nolimits_{\parC\ll\parA} \Pi^\parA_{\parC_{12}}$. 
If the two boxes removed from $\parA$ to obtain $\parC$ are in the same row or the same column,
then $\Pi^\parA_{\parC_{12}}$ projects onto the unique instance of the irrep $\parC\times\parA$ in $V$, and $\Pi^\parA_{\parC_{12}}=\Pi^\parA_{id,\parC_{12}}$
or $\Pi^\parA_{\parC_{12}}=\Pi^\parA_{sgn,\parC_{12}}$, respectively.
On the other hand,
if they are in different rows and columns,
then $\Pi^\parA_{\parC_{12}}= \Pi^\parA_{id,\parC_{12}} + \Pi^\parA_{sgn,\parC_{12}}$, where each $\Pi^\parA_{id,\parC_{12}}$ and $\Pi^\parA_{sgn,\parC_{12}}$ projects onto an instance of the irrep $\parC\times\parA$.
Hence, similarly to 
(\ref{eqn:GammaBlocks}), we can express $\Gamma_{1,2}$ as a linear combination
\begin{equation}\label{eqn:Gamma12Blocks}
\Gamma_{1,2} = \sum_{\parA\vdash N}
\Big(
\sum_{\parC\ll_c\parA}     \alpha^\parA_{id,\parC}\Pi^\parA_{id,\parC_{12}}  +
\sum_{\parC\ll_{rc}\parA} \alpha^\parA_{sgn,\parC} \Pi^\parA_{sgn,\parC_{12}\leftarrow id,\parC_{12}}
\Big).
\end{equation}

If $\parC\ll_{rc}\parA$, then there exist two distinct $\parB,\parB'\vdash N-1$ such that $\parC<\parB<\parA$ and $\parC<\parB'<\parA$, and let $\parB$ appear in the lexicographic order after $\parB'$. Note that $\PiA_{\parC_{12},\parB_1}$ projects onto a single instance of $\parC\times\parA$. 
We have \[
\PiA_{sgn,\parC_{12}\leftarrow id,\parC_{12}}
\propto
\PiA_{sgn,\parC_{12}}
\PiA_{\parC_{12},\parB_1}
\PiA_{id,\parC_{12}},
\]
and we specify the global phase of the transporter $\PiA_{sgn,\parC_{12}\leftarrow id,\parC_{12}}$ by assuming that the coefficient of this proportionality is positive. We present the value of this coefficient in \refsec{overlaps1}.

Let us relate (\ref{eqn:GammaBlocks}) and (\ref{eqn:Gamma12Blocks}), the two ways in which we can specify the adversary matrix.
 One can see that the $2(N-2)!\times N!$ submatrix of $\Xi^\parA_{id,\parC_{12}}$ and $\Xi^\parA_{sgn,\parC_{12}}$ corresponding to $\cD_{1,2}^{s,t}$ is proportional, respectively, to the $2(N-2)!\times N!$ submatrix of 
$\Pi_{id,\parC_{12}}^\parA$
and
$\Pi_{sgn,\parC_{12}\leftarrow id,\parC_{12}}^\parA$
corresponding to $\cD_{1,2}^{s,t}$.
Hence, just like in \refeqn{GammaViaBlocks}, we have
\[
 \Xi_{id,\parC}^\parA 
=
\frac{1}{\gamma_{id,\parC}^\parA} \sum_{\pi\in R} U_\pi \Pi_{id,\parC_{12}}^\parA  V_{\pi^{-1}}
\qquad
\text{and} 
\qquad
\Xi_{sgn,\parC}^\parA 
=
\frac{1}{\gamma_{sgn,\parC}^\parA} \sum_{\pi\in R} U_\pi \Pi_{sgn,\parC_{12}\leftarrow id,\parC_{12}}^\parA  V_{\pi^{-1}},
\]
and we specify the global phase of the transporters $\Xi$ by assuming that the normalization scalars $\gamma$ are positive.
Note that
\begin{multline*}
(\gamma_{id,\parC}^\parA)^2\Pi^\parA_\parA  =
(\gamma_{id,\parC}^\parA\Xi^\parA_{id,\parC})^*
(\gamma_{id,\parC}^\parA\Xi^\parA_{id,\parC}) = 
\Big(\sum_{\pi\in R} U_\pi \Pi_{id,\parC_{12}}^\parA  V_{\pi^{-1}}\Big)^*
\sum_{\pi\in R} U_\pi \Pi_{id,\parC_{12}}^\parA  V_{\pi^{-1}} \\
=
\binom{N}{2}\frac{1}{N!}\sum_{\pi\in\S_{[N]}} V_\pi \Pi_{id,\parC_{12}}^\parA  V_{\pi^{-1}}
=
\binom{N}{2} \frac{\dim\parC}{\dim\parA}\Pi^\parA_\parA,
\end{multline*}
where the last equality holds because $V_\pi$ and $\Pi^\parA$ commute (thus the sum has to be proportional to $\Pi^\parA_\parA$) and $\Tr[\Pi^\parA_{id,\parC_{12}}]\big/\Tr[\Pi^\parA_\parA]=\dim\parC/\dim\parA$. The same way we calculate $\gamma_{sgn,\parC}^\parA$, and we have
\[
\gamma_{id,\parC}^\parA=\frac{\beta_{id,\parC}^\parA}{\alpha_{id,\parC}^\parA} =
\gamma_{sgn,\parC}^\parA=\frac{\beta_{sgn,\parC}^\parA}{\alpha_{sgn,\parC}^\parA} = 
\sqrt{\binom{N}{2}\frac{\dim\parC}{\dim\parA}}.
\]

\section{\texorpdfstring{Tools for estimating $\|\Delta_1\circ\Gamma\|$}{Tools for estimating the norm of the product of the adversary matrix and the difference matrix}} \label{sec:tools}

\subsection{\texorpdfstring{Division of $\Delta_1\circ\Gamma$ into two parts}{Division of the adversary matrix into two parts}} \label{sec:TwoHalves}

For all $j\in[2,N]$, $\Delta_1\circ \Gamma_{1,j}$ is essentially the same as $\Delta_1\circ \Gamma_{1,2}$. And, for all $\{i,j\}\subset[2,N]$, $\Delta_1\circ \Gamma_{i,j}$ is essentially the same as $\Delta_1\circ \Gamma_{2,3}$, which, in turn, is essentially the same as $\Delta_3\circ \Gamma_{1,2}$. Let us distinguish these two cases by dividing $\Gamma$ into two parts: let 
$\Gamma'$ be the $(N-1)N!\times N!$ submatrix of $\Gamma$ corresponding to $x\in\cD_{1,j}$, where $j\in[2,N]$, and let $\Gamma''$ be the $\binom{N-1}{2}N!\times N!$ submatrix of $\Gamma$ corresponding to $x\in\cD_{i,j}$, where $\{i,j\}\in[2,N]$. 

\begin{clm} \label{clm:TwoHalves}
We have $\|\Delta_1\circ\Gamma\|\in\O(1)$ if and only if both $\|\Delta_1\circ\Gamma'\|\in\O(1)$ and $\|\Delta_1\circ\Gamma''\|\in\O(1)$.
\end{clm}

Let
$R'=\Rep(\S_{[2,N]}/\S_{[3,N]})$
and $R''=\Rep(\S_{[N]\setminus\{3\}}/(\S_{\{1,2\}}\times\S_{[4,N]}))$ be transversals of the left cosets of $\S_{[3,N]}$ in $\S_{[2,N]}$ and
of $\S_{\{1,2\}}\times\S_{[4,N]}$ in $\S_{[N]\setminus\{3\}}$, respectively.
Similarly to \refeqn{GammaViaBlocks}, we have
\begin{equation}
\label{eqn:GammaHalves}
\Delta_1\circ\Gamma'
=\sum_{\pi\in R'}
U_{\pi} (\Delta_1\circ\Gamma_{1,2}) V_{\pi^{-1}}
\qqAnd
\Delta_1\circ\Gamma''
=
U_{(13)}\Big(
\sum_{\pi\in R''}
U_{\pi} (\Delta_3\circ\Gamma_{1,2}) V_{\pi^{-1}}
\Big)V_{(13)},
\end{equation}
which imply
\begin{align} 
\label{eqn:Half1Norm} &
\big\|\Delta_1\circ\Gamma'\big\|^2 =
\big\|(\Delta_1\circ\Gamma')^*(\Delta_1\circ\Gamma')\big\|
=
\Big\|
\sum_{\pi\in R'}
V_{\pi} (\Delta_1\circ\Gamma_{1,2})^* (\Delta_1\circ\Gamma_{1,2}) V_{\pi^{-1}}
\Big\|,
\\ \label{eqn:Half2Norm} &
\big\|\Delta_1\circ\Gamma''\big\|^2 =
\big\|(\Delta_1\circ\Gamma'')^*(\Delta_1\circ\Gamma'')\big\|
=
\Big\|
\sum_{\pi\in R''}
V_{\pi} (\Delta_3\circ\Gamma_{1,2})^* (\Delta_3\circ\Gamma_{1,2}) V_{\pi^{-1}}
\Big\|.
\end{align}
Therefore, we have to consider $\Delta_1\circ\Gamma_{1,2}$ and $\Delta_3\circ\Gamma_{1,2}$.

\subsection{\texorpdfstring
{Commutativity with the action of $\Delta_i$}
{Commutativity with the action of the difference matrix}
} \label{sec:DeltaComm}

Instead of $\Delta_i$, let us first consider the action of $\overline\Delta_i$. For
$i\in[N]$ and $s\in\Sigma$, let $\hat\Pi_i^s$ be the projector on all $y\in\cD_0$ such that $y_i=s$. Then, due to the particular way we define the bijection $f$, we have
\begin{equation}
\label{eqn:DeltaViaProj}
\overline\Delta_i\circ\Gamma_{1,2}= \sum\nolimits_{s\in\Sigma}\hat\Pi^s_i \Gamma_{1,2} \hat\Pi^s_i
\quad\text{whenever}\quad i\neq 2
\qquad\text{and}\qquad
\overline\Delta_2\circ\Gamma_{1,2}= \sum\nolimits_{s\in\Sigma}\hat\Pi^s_1\Gamma_{1,2}\hat\Pi^s_2.
\end{equation}
Note that $\hat\Pi_i^s$ commutes with every $\Pi_{\rho_{j_1\!\ldots j_m}}$ whenever $i\in\{j_1,\ldots,j_m\}$.
Hence, for $i\in\{j_1,\ldots,j_m\}\setminus\{2\}$ and every $N!\times N!$ matrix $A$, we have 
\begin{equation}
\label{eqn:DeltaComm}
  \Delta_i\circ(\Pi_{\rho_{j_1\!\ldots j_m}}A) =\Pi_{\rho_{j_1\!\ldots j_m}}(\Delta_i\circ A)
\qquad\text{and}\qquad
  \Delta_i\circ(A\Pi_{\rho_{j_1\!\ldots j_m}}) =(\Delta_i\circ A)\Pi_{\rho_{j_1\!\ldots j_m}}.
\end{equation}

\subsection{\texorpdfstring
{Relations among irreps of $\S_{[3,N]}\times\S_\Sigma$ within an isotypical subspace}
{Relations among irreps of S\_[3,N] x S\_Sigma within an isotypical subspace} }\label{sec:overlaps1}

We are interested to see how $\Delta_1$ acts on $\Gamma_{1,2}$, which 
 requires us to consider how it acts on $\Pi^\parA_{id,\parC_{12}}$ and $\Pi^\parA_{sgn,\parC_{12}\leftarrow id,\parC_{12}}$. Unfortunately, this action is hard to calculate directly, therefore we express $\Pi^\parA_{id,\parC_{12}}$ and $\Pi^\parA_{sgn,\parC_{12}\leftarrow id,\parC_{12}}$ as linear combinations of certain operators on which the action of $\Delta_1$  is easier to calculate.

Consider $\parA\vdash N$ and $\parC\ll_{rc}\parA$.
The projector $\Pi^\parA_{\parC_{12}}$ projects onto the isotypical subspace of $\cY$ corresponding to the irrep $\parC\times\parA$ of $\S_{[3,N]}\times\S_\Sigma$, and this subspace contains two instances of this irrep. 
There are as many degrees of freedom in splitting this subspace in half so that each half  corresponds to a single instance of the irrep as in splitting $\R^2$ in orthogonal one-dimensional subspaces.
We already considered one such split, 
$\Pi^\parA_{\parC_{12}}=\Pi^\parA_{id,\parC_{12}}+\Pi^\parA_{sgn,\parC_{12}}$,
 and now let us relate it to another.

Let $\parB,\parB'\vdash N-1$ be such that $\parC<\parB<\parA$, $\parC<\parB'<\parA$, and $\parB$ appears after $\parB'$ in the lexicographical order.
Then $\Pi^\parA_{\parC_{12},\parB_{1}}$ and $\Pi^\parA_{\parC_{12},\parB'_{1}}$ project onto two orthogonal instances of the irrep $\parC\times\parA$, and 
$\Pi^\parA_{\parC_{12}}=\Pi^\parA_{\parC_{12},\parB_{1}}+\Pi^\parA_{\parC_{12},\parB'_{1}}$.
Note that $V_{(12)}$ commutes with $\Pi^\parA_{\parC_{12}}$ and that $\Pi^\parA=\Pi_\parA$. The {\em orthogonal form} \cite[Section 3.4]{james:symmetric} of the irrep $\parA$ tells us that $V_{(12)}$ restricted to the isotypical subspace corresponding to $\parC\times\parA$ is 
\begin{equation} \label{eqn:V12rest}
V_{(12)}\big|_{\parC_{12}\times\parA} = \frac{1}{d_{\parA,\parC}}
\Big( \Pi^\parA_{\parC_{12},\parB'_{1}} - \Pi^\parA_{\parC_{12},\parB_{1}}
+ 
\sqrt{d^2_{\parA,\parC}-1}\,
\Pi^\parA_{\parC_{12},\parB'_{1}\leftrightarrow\parC_{12},\parB_{1}}
\Big).
\end{equation}
Expression \refeqn{V12rest}, in effect, defines the global phase of transporters $\Pi^\parA_{\parC_{12},\parB'_{1}\leftarrow\parC_{12},\parB_{1}}$ and $\Pi^\parA_{\parC_{12},\parB'_{1}\leftarrow\parC_{12},\parB_{1}}$.

Recall that $\Pi_{id}=(\I+V_{(12)})/2$, and therefore
\begin{equation}
\label{eqn:PiId}
\Pi^\parA_{id,\parC_{12}}
= \frac{\Pi^\parA_{\parC_{12}}+V_{(12)}\big|_{\parC_{12}\times\parA}}{2}
= 
\frac{d_{\parA,\parC}-1}{2d_{\parA,\parC}}
\Pi^\parA_{\parC_{12},\parB_{1}} 
+ \frac{d_{\parA,\parC}+1}{2d_{\parA,\parC}}
\Pi^\parA_{\parC_{12},\parB'_{1}}
+ 
\frac{\sqrt{d^2_{\parA,\parC}-1}}{2d_{\parA,\parC}}
\Pi^\parA_{\parC_{12},\parB'_{1}\leftrightarrow\parC_{12},\parB_{1}}
\end{equation}
and
\begin{multline}
\label{eqn:PiSgn}
\Pi^\parA_{sgn,\parC_{12}\leftarrow id,\parC_{12}}
=\frac{2 d_{\parA,\parC}}{\sqrt{d^2_{\parA,\parC}-1}}
\Pi^\parA_{sgn,\parC_{12}} 
\Pi^\parA_{\parC_{12},\parB_{1}}
\Pi^\parA_{id,\parC_{12}} \\
 =
\frac{\sqrt{d_{\parA,\parC}^2-1}}{2d_{\parA,\parC}}
\Pi^\parA_{\parC_{12},\parB_{1}} 
- \frac{\sqrt{d_{\parA,\parC}^2-1}}{2d_{\parA,\parC}}
\Pi^\parA_{\parC_{12},\parB'_{1}}
+ 
\frac{d_{\parA,\parC}+1}{2d_{\parA,\parC}}
\Pi^\parA_{\parC_{12},\parB_{1}\leftarrow\parC_{12},\parB'_{1}}
-
\frac{d_{\parA,\parC}-1}{2d_{\parA,\parC}}
\Pi^\parA_{\parC_{12},\parB'_{1}\leftarrow\parC_{12},\parB_{1}}.
\end{multline}

\newcommand\dd{d_{\parEb,\parFb_{12}}} 

\subsection{\texorpdfstring
{Relations among irreps of $\S_{[4,N]}\times\S_\Sigma$ within an isotypical subspace}
{Relations among irreps of S\_[4,N] x S\_Sigma within an isotypical subspace}
}\label{sec:overlaps2}

We are also interested to see how $\Delta_3$ acts on $\Gamma_{1,2}$, which will require us to consider irreps of $\S_{[4,N]}\times\S_\Sigma$.
Let us now consider $k\in o(N)$, $\parE\vdash k$, and $\parF<\parE$. Recall that, according to our notation, $\parEB=(N-k,\parE)\vdash N$ and $\parFB_{123}=(N-k-2,\parF)_{123}\vdash N-3$ is obtained from $\parEB$ by removing two boxes in the first row and one box below the first row.

$V$ contains three instances of the irrep $\parFB_{123}\times\parEB$ of $\S_{[4,N]}\times\S_\Sigma$:
we have
\[
\begin{split}
\PiEb_{\parFB_{123}} 
& =
\PiEb_{\parFB_{123},\parEB_{12},(\parEB_1)}+ 
\PiEb_{\parFB_{123},\parFB_{12},\parEB_{1}}+
\PiEb_{\parFB_{123},(\parFB_{12}),\parFB_{1}}
=
\PiEb_{id,\parFB_{123},\parEB_{3}}+ 
\PiEb_{sgn,\parFB_{123},\parEB_{3}}+
\PiEb_{(id),\parFB_{123},\parFB_{3}},
\end{split}
\]
where each projector (other than $\PiEb_{\parFB_{123}}$) projects on a single instance of the irrep and the subscripts in parenthesis are optional. These
two decompositions follow essentially the chain of restrictions
$
\S_{[N]}\rightarrow \S_{[2,N]}\rightarrow \S_{[3,N]}\rightarrow \S_{[4,N]}
$ 
and $\S_{[N]}\rightarrow \S_{[N]\setminus\{3\}}\rightarrow \S_{\{1,2\}}\times\S_{[4,N]}\rightarrow \S_{[4,N]},
$
respectively.

From the orthogonal form of the irrep $\parEB$, we get that the restriction of $V_{(12)}$ and $V_{(23)}$ to the isotypical subspace corresponding to $\parFB_{123}\times\parEB$ is, respectively,
\begin{align*}
&
V_{(12)}\big|_{\parFB_{123}\times\parEB}
= \PiEb_{\parFB_{123},\parEB_{12}}
+ \frac{1}{\dd}
\Big(
\PiEb_{\parFB_{123},\parFB_{12},\parEB_{1}}
-  \PiEb_{\parFB_{123},\parFB_{1}}
+\sqrt{\dd^2-1}\,
\PiEb_{\parFB_{123},\parFB_{12},\parEB_{1}\leftrightarrow\parFB_{123},\parFB_{1}}
\Big),
\\&
V_{(23)}\big|_{\parFB_{123}\times\parEB}
= \frac{1}{\dd-1}
\Big(
\PiEb_{\parFB_{123},\parEB_{12}}
-
\PiEb_{\parFB_{123},\parFB_{12},\parEB_{1}}
+\sqrt{(\dd-1)^2-1}\,
\PiEb_{\parFB_{123},\parEB_{12}\leftrightarrow\parFB_{123},\parFB_{12},\parEB_{1}}
\Big)
+ \PiEb_{\parFB_{123},\parFB_{1}},
\end{align*}
where the global phases of the transporters in the expression for $V_{(12)}\big|_{\parFB_{123}\times\parEB}$ are consistent with \refeqn{V12rest}.
Therefore we can calculate the ``overlap'' of $\PiEb_{\parFB_{123},\parEB_{12}}$ and
\[
\PiEb_{id,\parFB_{123},\parEB_{3}}
=
V_{(13)} \big(\I+V_{(23)}\big)\PiEb_{\parFB_{123},\parEB_{1}} V_{(13)} \big/2
=
V_{(23)}V_{(12)}
\big(\I+V_{(23)}\big)
\big(\PiEb_{\parFB_{123},\parEB_{12}}+\PiEb_{\parFB_{123},\parFB_{12},\parEB_{1}}\big)
V_{(12)} V_{(23)}
\big/2
\]
to be
\begin{equation} \label{eqn:overlap}
\frac{\Tr\big[
\PiEb_{\parFB_{123},\parEB_{12}}
\PiEb_{id,\parFB_{123},\parEB_{3}}
\big]}{\dim\parFB_{123}\dim\parEB}
=\frac{2}{\dd(\dd-1)}.
\end{equation}
Since $\PiEb_{\parFB_{123},\parEB_{12}}=\Pi_{id}\PiEb_{\parFB_{123},\parEB_{12}}$,
we have 
\begin{equation}
\label{eqn:ThetaDec1}
\PiEb_{\parFB_{123},\parEB_{12}} =
\PiEb_{\parFB_{123},\parFB_{3}}
+ \frac{2}{\dd^2-\dd}
\Big(
\PiEb_{id,\parFB_{123},\parEB_{3}} - \PiEb_{\parFB_{123},\parFB_{3}}
\Big)
+\frac{\sqrt{2\big(\dd^2-\dd-2\big)}}{\dd^2-\dd}
\PiEb_{\parFB_{123},\parFB_{3}\leftrightarrow id,\parFB_{123},\parEB_{3}}.
\end{equation}

\subsection{\texorpdfstring
{Summing the permutations of $(\Delta_1\circ\Gamma_{1,2})^*(\Delta_1\circ\Gamma_{1,2})$}
{Summing the permutations}
}

We will express $(\Delta_1\circ\Gamma_{1,2})^*(\Delta_1\circ\Gamma_{1,2})$ as a linear combination of projectors $\PiA_{\parC_{12},\parB_1}$ and transporters $\PiA_{\parC_{12},\parB'_1\leftarrow\parC_{12},\parB_1}$,
where $\parA\vdash N$, $\parC\ll_c\parA$, and $\parB,\parB'\vdash N-1$ are such that $\parC<\parB\!<\parA$ and $\parC<\parB'\!<\parA$ (we consider transporters only if $\parC\ll_{rc}\parA$, and thus $\parB\neq\parB'$). In order to calculate $\|\Delta_1\circ\Gamma'\|$ via \refeqn{Half1Norm}, we use
\begin{equation}
\label{eqn:PermSum1}
\begin{split} &
\frac{1}{N-1}  \sum_{\pi\in R'}
V_{\pi} \PiA_{\parC_{12},\parB_1} V_{\pi^{-1}}
=
\frac{1}{(N-1)!}  \sum_{\pi\in\S_{[2,N]}}
V_{\pi} \PiA_{\parC_{12},\parB_1} V_{\pi^{-1}}
=
\frac{\Tr\big[
V_{\pi} \PiA_{\parC_{12},\parB_1} V_{\pi^{-1}} 
\big]}
{\Tr\big[ \PiA_{\parB_1} \big]}
\PiA_{\parB_1}
=
\frac{\dim\parC}
{\dim\parB}
\PiA_{\parB_1},
\\ &
\frac{1}{N-1}  \sum_{\pi\in R'}
V_{\pi} \Pi^\parA_{\parC_{12},\parB'_1\leftarrow\parC_{12},\parB_1} V_{\pi^{-1}}
=
\frac{1}{(N-1)!}  \sum_{\pi\in\S_{[2,N]}}
V_{\pi} \Pi^\parA_{\parC_{12},\parB'_1\leftarrow\parC_{12},\parB_1} V_{\pi^{-1}}
=
0.
\end{split}
\end{equation}
The equalities in \refeqn{PermSum1} hold because, first of all, 
$\PiA_{\parC_{12},\parB_1}$ and $\Pi^\parA_{\parC_{12},\parB'_1\leftarrow\parC_{12},\parB_1}$ are fixed under $\S_{[3,N]}\times\S_\Sigma$.
Second, $V$ as a representation of $\S_{[2,N]}\times\S_\Sigma$ is {\em multiplicity-free} (i.e., it contains each irrep at most once), and thus every operator on $\cY$ that is fixed under $\S_{[2,N]}\times\S_\Sigma$ can be expressed as a linear combination of projectors $\Pi^{\parA'}_{\parB''_1}$, where $\parA'\vdash N$ and $\parB''<\parA'$.
And third, for $\pi\in\S_{[2,N]}$, $V_\pi$ commutes with both $\PiA_{\parB_1}$ and $\PiA_{\parB'_1}$.

\section{Construction of the optimal adversary matrix}
\label{sec:sufficient}

In \refsec{Gamma12ViaProj} we showed that
\(
\beta_{id,\parC}^\parA/\alpha_{id,\parC}^\parA
=
\beta_{sgn,\parC}^\parA/\alpha_{sgn,\parC}^\parA
=
\sqrt{\binom{N}{2}\frac{\dim\parC}{\dim\parA}}
\).
We calculate $\dim\parC$ and $\dim\parA$ using the hook-length formula, and one can see that, given a fixed $\parD\vdash k$, $\dim\parDb$ can be expressed as a polynomial in $N$ of  degree $k$ and having the leading coefficient $1/h(\parD)$ (see \refeqn{dimParD}). Therefore we get that \refclm{GammaTransp} is equivalent to the following claim, which we prove in \refapp{necessary}.
\begin{clm}
\label{clm:alphas}
Suppose $\Gamma_{1,2}$ is given as in (\ref{eqn:Gamma12Blocks}), $\alpha^{(N)}_{id,(N-2)}=N^{-1/3}$, and $\Gamma$ is obtained from $\Gamma_{1,2}$ via \refeqn{GammaViaBlocks}. 
Consider $\parA\vdash N$ that has $\O(1)$ boxes below the first row and $\parC\ll_c\parA$.
In order for $\|\Delta_1\circ\Gamma\|\in\O(1)$ to hold,  
 we need to have 
\begin{enumerate}
\item \label{pt1}
$\alpha^\parA_{id,\parC}=N^{-1/3}+\O(1/N)$ if $\parA$ and $\parC$ are the same below the first row,
\item \label{pt2}
$\alpha^\parA_{id,\parC},\alpha^\parA_{sgn,\parC}=N^{-1/3}+\O(1/\sqrt{N})$ if $\parA$ has one box more below the first row than $\parC$,
\item \label{pt3}
$\alpha^\parA_{id,\parC},\alpha^\parA_{sgn,\parC}=\O(1)$ if $\parA$ has two boxes more below the first row than $\parC$.
\end{enumerate}
(Note that $\alpha^{(N)}_{id,(N-2)}=N^{-1/3}$ implies $\|\Gamma\|\geq\beta^{(N)}_{id,(N-2)}\in\Theta(N^{2/3})$.)
\end{clm}

Consider $k\in o(N)$ and $\parE\vdash k$. Claims \ref{clm:GammaTransp} and \ref{clm:alphas} hint that for the optimal adversary matrix we could choose coefficients $\alpha_{id,\parEb_{12}}^\parEb\approx \alpha_{id,\parEb_{12}}^\parDb\approx \alpha_{sgn,\parEb_{12}}^\parDb$ whenever $\parD>\parE$ and $\alpha_{id,\parEb_{12}}^\parDb= \alpha_{sgn,\parEb_{12}}^\parDb=0$ whenever $\parD\gg\parE$. Let us do that. For $\parD>\parE$, note that $\parEB_{12}<\parEB_1<\parDB$, $\parEB_{12}<\parDB_1<\parDB$, and $\parEB_1$ appears after $\parDB_1$ in the lexicographic order, and also note that $d_{\parDB,\parEB_{12}}\geq N-2k-1$ (equality is achieved by $\parE=(k)$ and $\parD=(k+1)$). Therefore, according to \refeqn{PiId} and \refeqn{PiSgn}, we have
\begin{multline*}
\Pi^\parEb_{id,\parEb_{12}}+
\sum_{\parD>\parE}\big(
\Pi^\parDB_{id,\parEB_{12}} +
\Pi^\parDB_{sgn,\parEB_{12}\leftarrow id,\parEB_{12}}
\big) =
\Pi^\parEb_{\parEb_{12}}+
\sum_{\parD>\parE}\big(
\Pi^\parDB_{\parEB_{12},\parEB_1} +
\Pi^\parDB_{\parEB_{12},\parEB_1\leftarrow \parEB_{12},\parDB_1}
\big)
 + \O(1/N) \\
=
\Pi^\parEb_{\parEb_{12}}+
\sum_{\parD>\parE}
2\Pi^\parDB_{\parEB_{12},\parEB_1}\Pi_{id} + \O(1/N)
=
2\Pi_{\parEB_{12},\parEB_1}\Pi_{id}
- \Pi^\parEb_{\parEb_{12}} + \O(1/N),
\end{multline*}
where the last equality is due to $\Pi^\parEB_{\parEB_{12}}=\Pi^\parEB_{\parEB_{12},\parEB_1}=\Pi^\parEB_{id,\parEB_{12}}$ and 
\(
\Ind\nolimits_{\,\S_{[2,N]}}^{\,\S_{[N]}} \!\parEB_1 \cong \parEB\oplus\bigoplus\nolimits_{\parD>\parE}\parDB
\),
that is, the branching rule.
Thus we choose to construct $\Gamma_{1,2}$ as a linear combination of matrices
\[
2\Pi_{\parEB_{12},\parEB_1}\Pi_{id}
- \Pi^\parEb_{\parEb_{12}}
= \Pi^\parEb_{\parEb_{12}}
+
\sum_{\parD>\parE}
\bigg(
\frac{d_{\parDB,\parEB_{12}}-1}{d_{\parDB,\parEB_{12}}}
 \Pi^\parDB_{id,\parEB_{12}} +
\frac{\sqrt{d_{\parDB,\parEB_{12}}^2-1}}{d_{\parDB,\parEB_{12}}}
\, \Pi^\parDB_{sgn,\parEB_{12}\leftarrow id,\parEB_{12}}
\bigg).
\]
(At first glance, it may seem that the matrix on the left hand side does not ``treat'' indices $1$ and $2$ equally, but that is an illusion due to the way we define the bijection $f$.)

\begin{thm}
\label{thm:main}
Let $\Gamma$ be constructed via \refeqn{GammaViaBlocks} from
\[
  \Gamma_{1,2} = \sum_{k=0}^{N^{2/3}}
 \frac{N^{2/3}-k}{N}
   \sum_{\parE\vdash k} (2\Pi_{\parEB_{12},\parEB_1}\Pi_{id} - \Pi^\parEB_{\parEB_{12}}). 
\]
Then $\|\Gamma\|\in \Omega(N^{2/3})$ and $\|\Delta_1\circ\Gamma\|\in\O(1)$, and therefore $\Gamma$ is, up to constant factors, an optimal adversary matrix for {\sc Element Distinctness}.
\end{thm}

For $\Gamma_{1,2}$ of \refthm{main} expressed in the form \refeqn{Gamma12Blocks}, we have $\alpha^{(N)}_{id,(N-2)}=N^{-1/3}$, and 
 therefore $\|\Gamma\|\in \Omega(N^{2/3})$. In the remainder of the paper, let us prove $\|\Delta_1\circ\Gamma'\|\in\O(1)$ and $\|\Delta_1\circ\Gamma''\|\in\O(1)$,
 which is sufficient due to \refclm{TwoHalves}.

\subsection{\texorpdfstring
{Approximate action of $\Delta_i$}
{Approximate action of the difference matrix}}

The precise calculation of $\Delta_1\circ\Gamma$ is tedious; we consider it in \refapp{necessary}. Here, however, it suffices to upper bound
 $\norm|\Delta_1\circ\Gamma|$ using the following trick first introduced in~\cite{belovs:adv-el-dist} and later used in~\cite{spalek:kSumLower,belovs:nonAdaptiveLG,spalek:adv-array,rosmanis:collision}.
 
For any matrix $A$ of the same dimensions as $\Delta_i$, we call a matrix $B$ satisfying
 \(
  \Delta_i\circ B=\Delta_i\circ A 
 \) 
 an {\em approximation} of $\Delta_i\circ A$ and we denote it with
 $\Delta_i\diamond A$.
From the fact \refeqn{gamma2} on the $\gamma_2$ norm, it follows that 
\(
  \norm|\Delta_i\circ A| 
  \leq 2\norm|\Delta_i\diamond A|.
\)
Hence, to show that $\|\Delta_1\circ\Gamma'\|\in \O(1)$ and $\|\Delta_1\circ\Gamma''\|\in \O(1)$, it suffices to show that
 $\|\Delta_1\diamond\Gamma'\|\in \O(1)$ and $\|\Delta_1\diamond\Gamma''\|\in \O(1)$ for any $\Delta_1\diamond\Gamma'$ and $\Delta_1\diamond\Gamma''$.
 That is, it suffices to show that we
 can change entries of $\Gamma'$ and $\Gamma''$ corresponding to $(x,y)$ with $x_1=y_1$ in a way that the 
 spectral norms of the resulting matrices are constantly bounded.

Note that we can always choose $\Delta_i\diamond A = A$
 and
\(
 \Delta_i\diamond(A+A')
  =\Delta_i\diamond A + \Delta_i\diamond A'.
\)
We will express $\Gamma_{1,2}$ as a linear combination of certain $N!\times N!$ matrices and, for every such matrix $A$, we will choose  $\Delta_i\diamond A = A$, except for the following three, for which we calculate the action of $\Delta_1$ or $\Delta_3$ precisely. We have
\[
\Delta_1\circ\Pi_{id}=V_{(12)}/2 ,\qquad
\Delta_3\circ\Pi_{\parFb_{123},\parFb_{3}} = 0,\qqAnd
\Delta_3\circ\Pi_{\parFb_{123},\parFb_{13}} = 0 
\]
due to $\Delta_1\circ\I=\Delta_3\circ\I=0$ and the commutativity relation \refeqn{DeltaComm}.

Due to \refeqn{DeltaComm}, we also have $\Delta_3\circ(A\Pi_{id})=(\Delta_3\circ A)\Pi_{id}$ for every $N!\times N!$ matrix $A$. One can see that, given any choice of $\Delta_3\diamond A$, we can choose $\Delta_3\diamond(A\Pi_{id})=(\Delta_3\diamond A)\Pi_{id}$.

\subsection{\texorpdfstring
{Bounding $\|\Delta_1\circ\Gamma'\|$}
{Bounding the first part}}

For $k\leq N^{2/3}$ and $\parE\vdash k$, define $N!\times N!$ matrices $(\Gamma_\parE)_{1,2}$ and $(\Gamma_k)_{1,2}$ such that 
\[
\Gamma_{1,2}=\sum_{k=0}^{N^{2/3}}
\frac{N^{2/3}-k}{N}(\Gamma_k)_{1,2},
\qquad
(\Gamma_k)_{1,2}=\sum_{\parE\vdash k} (\Gamma_\parE)_{1,2},
\qqAnd
(\Gamma_\parE)_{1,2}
=2\Pi_{\parEb_{12},\parEb_1}\Pi_{id} - \Pi^\parEb_{\parEb_{12}}.
\]
The projector $\Pi_{\parEb_{12},\parEb_{1}}$ commutes with the action of $\Delta_1$, therefore
we can choose
\begin{multline*}
\Delta_1\diamond (\Gamma_\parE)_{1,2}
 =
2\Pi_{\parEb_{12},\parEb_1}(\Delta_1\circ \Pi_{id}) - \Pi^\parEb_{\parEb_{12}}
 = 
 \Pi_{\parEb_{12},\parEb_1}V_{(12)}  -\PiEb_{\parEb_{12}} \\
 = \sum_{\parD>\parE} \Pi^\parDb_{\parEb_{12},\parEb_1} V_{(12)} 
 = \sum_{\parD>\parE} \bigg(
-\frac{1}{d_{\parDb,\parEb_{12}}}\Pi^\parDb_{\parEb_{12},\parEb_1}
+\frac{\sqrt{d_{\parDb,\parEb_{12}}^2-1}}{d_{\parDb,\parEb_{12}}}\Pi^\parDb_{\parEb_{12},\parEb_1\leftarrow\parEb_{12},\parDb_1}
\bigg),
\end{multline*}
where the third equality is due to the branching rule and both $\PiEb_{\parEb_{12}}=\PiEb_{\parEb_{12}}\Pi_{id}$ and $\Pi_{id}V_{(12)}=\Pi_{id}$,
and the last equality comes from \refeqn{V12rest}.
To estimate the norm of $\Delta_1\diamond\Gamma'$ via 
 \refeqn{Half1Norm}, we have
\begin{align}
 & \notag
  \sum_{\pi\in R'}V_\pi(\Delta_1\diamond (\Gamma_\parE)_{1,2})^*(\Delta_1\diamond (\Gamma_\parE)_{1,2})V_{\pi^{-1}}
\\[-5pt] & \hspace{25pt} \notag
\preceq
 \sum_{\parD>\parE}
 \sum_{\pi\in R'}V_\pi
 \Big(
\frac{1}{d_{\parDb,\parEb_{12}}^2}
\Pi^\parDb_{\parEb_{12},\parEb_1}
+
\Pi^\parDb_{\parEb_{12},\parDb_1}
-\frac{\sqrt{d_{\parDb,\parEb_{12}}^2-1}}{d_{\parDb,\parEb_{12}}^2}
\Pi^\parDb_{\parEb_{12},\parEb_1\leftrightarrow\parEb_{12},\parDb_1} 
\Big)
 V_{\pi^{-1}}
\\ & \hspace{50pt} \notag
=  
(N-1)
\sum_{\parD>\parE} \Big(
\frac{1}{d_{\parDb,\parEb_{12}}^2}
\frac{\dim\parEb_{12}}{\dim\parEb_{1}}
\Pi^\parDb_{\parEb_1}
+ 
\frac{\dim\parEb_{12}}{\dim\parDb_{1}}
\Pi^\parDb_{\parDb_1}
\Big)
\\ & \hspace{75pt} \label{eqn:GammaLambda}
\preceq
\frac{1}{N-o(N)} 
\sum_{\parD>\parE} \Pi^\parDb_{\parEb_1}
+ 
(N-1)
\sum_{\parD>\parE}
\frac{\dim\parEb_{12}}{\dim\parDb_{1}}
\Pi^\parDb_{\parDb_1},
\end{align}
where  $\preceq$ denotes the {\em semidefinite ordering}, the equality in the middle comes from \refeqn{PermSum1}, and the last inequality is due to $\dim\parEb_{12}\leq \dim\parEb_{1}$ and $d_{\parDb,\parEb_{12}}\geq N-2k-1$.

\begin{clm} \label{clm:dimFraction} Let $\parD\vdash k$. Then
\(
1-\dim\parDb_1/\dim\parDb\leq 2k/N.
\)
\end{clm}

\begin{proof}
Recall the hook-length formula \refeqn{hook}. 
As $\parD$ has $\parD(1)\leq k$ columns, define $\parD^\top(j)=0$ for all $j\in[\parD(1)+1,k]$. We have
\begin{equation}\label{eqn:dimParD}
\dim\parDb = \frac{N!}{h((N-k,\parD))}
=\frac{N!/(N-2k)!}{h(\parD)\prod_{j=1}^k(N-k+1-j+\parD^\top(j))},
\end{equation}
and therefore
\[
1-\frac{\dim\parDb_1}{\dim\parDb} = 1-\frac{(N-1)!/(N-2k-1)!}{N!/(N-2k)!}
\prod_{j=1}^k\frac{N-k+1-j+\parD^\top(j)}{N-k-j+\parD^\top(j)}  < 1-\frac{N-2k}{N} = \frac{2k}{N}.
\]
\end{proof}

 For $\parE'\neq \parE$, we have
\(
(\Delta_1\diamond(\Gamma_{\parE'})_{1,2})^*(\Delta_1\diamond(\Gamma_\parE)_{1,2})=0,
\)
therefore, by summing \refeqn{GammaLambda} over all $\parE\vdash k$, we get
\begin{align}
\notag
&
  \sum_{\pi\in R'}V_\pi(\Delta_1\diamond (\Gamma_k)_{1,2})^*(\Delta_1\diamond (\Gamma_k)_{1,2})V_{\pi^{-1}}
\\ & \hspace{25pt}\notag
\preceq
\frac{1}{N-o(N)}
\sum_{\parE\vdash k}\sum_{\parD>\parE} \Pi^\parDb_{\parEb_1}
+ 
(N-1)
\sum_{\parD\vdash k+1}
\sum_{\parE<\parD}
\frac{\dim\parEb_{12}}{\dim\parDb_{1}}
\Pi^\parDb_{\parDb_1}
\\ & \hspace{50pt} \label{eqn:GammaK}
\preceq
\frac{1}{N-o(N)}
\sum_{\parE\vdash k}\sum_{\parD>\parE} \Pi^\parDb_{\parEb_1}
+ 
2(k+1)
\sum_{\parD\vdash k+1}
\Pi^\parDb_{\parDb_1},
\end{align}
where the first inequality holds because $\sum_{\parE\vdash k}\sum_{\parD>\parE}$ and $\sum_{\parD\vdash k+1}\sum_{\parE<\parD}$ 
are sums over the same pairs of $\parE$ and $\parD$,
and the second inequality holds because $\dim\parDb_1=\dim\parDb_{12}+\sum_{\parE<\parD}\dim\parEb_{12}$ (due to the branching rule) and \refclm{dimFraction}.

Finally, by summing \refeqn{GammaK} over $k$, we get
\begin{multline}\label{eqn:GammaKMult}
(\Delta_1\diamond\Gamma')^*(\Delta_1\diamond\Gamma') =
  \sum_{\pi\in R'}V_\pi(\Delta_1\diamond \Gamma_{1,2})^*(\Delta_1\diamond \Gamma_{1,2})V_{\pi^{-1}}
\\
\preceq
\sum_{k=0}^{N^{2/3}} \frac{(N^{2/3}-k)^2}{N^2}\bigg(
\frac{1}{N-o(N)}
\sum_{\parE\vdash k}\sum_{\parD>\parE} \Pi^\parDb_{\parEb_1}
+ 
2(k+1)
\sum_{\parD\vdash k+1}
\Pi^\parDb_{\parDb_1}
\bigg)\preceq\I/3.
\end{multline}
Hence, $\|\Delta_1\circ\Gamma'\|\in\O(1)$.
(Note: the norm of \refeqn{GammaK} is $\Theta(k)$ and, in \refeqn{GammaKMult}, we essentially multiply it with $\cT^2/N^2$, where $\cT$ is the intended lower bound. 
This provides an intuition for why one cannot prove a lower bound higher than $\Omega(N^{2/3})$.)

\subsection{\texorpdfstring
{Bounding $\|\Delta_1\circ\Gamma''\|$}
{Bounding the second part}}

Let us decompose the adversary matrix as $\Gamma=2\Gamma_\cA-\Gamma_\cB$, where we define $\Gamma_\cA$ and $\Gamma_\cB$ via their restriction to the rows labeled by $x\in\cD_{1,2}$:  
\[
  (\Gamma_\cA)_{1,2} = \sum_{k=0}^{N^{2/3}}
 \frac{N^{2/3}-k}{N}
   \sum_{\parE\vdash k} \Pi_{\parEb_{12},\parEb_1}\Pi_{id} 
\qqAnd
  (\Gamma_\cB)_{1,2} = \sum_{k=0}^{N^{2/3}}
 \frac{N^{2/3}-k}{N}
   \sum_{\parE\vdash k}  \Pi^\parEb_{\parEb_{12}},
\]
respectively. We show that $\|\Delta_1\circ\Gamma''_\cA\|\in\O(1)$ and $\|\Delta_1\circ\Gamma''_\cB\|\in\O(1)$, which together imply $\|\Delta_1\circ\Gamma''\|\in\O(1)$.
The argument is very similar for both $\Gamma_\cA$ and $\Gamma_\cB$, and let us start by showing $\|\Delta_1\circ\Gamma''_\cA\|\in\O(1)$.

We are interested to see how $\Delta_3$ acts on $(\Gamma_\cA)_{1,2}$. Let $\parF<\parE$, and
we will have to consider $\Pi_{\parFb_{123},\parEb_{12},\parEb_1}$.
For every $\parA>\parEb_1$, note that $V_{(23)}$ and $\PiA_{\parFb_{123},\parEb_1}$ commute. So, similarly to \refeqn{V12rest}, we have
\[
V_{(23)}\Pi_{\parFb_{123},\parEb_1}
= \frac{1}{d_{\parEb_1,\parFb_{123}}}\sum_{\parA>\parEb_1}\Big(
\PiA_{\parFb_{123},\parEb_{12},\parEb_1}
-
\PiA_{\parFb_{123},\parFb_{12},\parEb_1}
+
\sqrt{d_{\parEb_1,\parFb_{123}}^2-1}  \,
\PiA_{\parFb_{123},\parEb_{12},\parEb_1\leftrightarrow\parFb_{123},\parFb_{12},\parEb_1}
\Big).
\]
Hence
\[
\frac{\Tr{[
\PiA_{\parFb_{123},\parEb_{12},\parEb_1}
\PiA_{\parFb_{123},\parEb_{13},\parEb_1}
]}}{\dim\parFb_{123}\dim\parA}
=
\frac{\Tr{[
\PiA_{\parFb_{123},\parEb_{12},\parEb_1} V_{(23)}
\PiA_{\parFb_{123},\parEb_{12},\parEb_1} V_{(23)}
]}}{\dim\parFb_{123}\dim\parA}
=\frac{1}{d_{\parEb_1,\parFb_{123}}^2},
\]
and therefore, similarly to \refeqn{ThetaDec1}, we have
\begin{multline} \label{eqn:ThetaDec2}
\Pi_{\parFb_{123},\parEb_{12},\parEb_1}
=
\Pi_{\parFb_{123},\parFb_{13},\parEb_1}
+
\frac{1}{d_{\parEb_1,\parFb_{123}}^2}
\big(
\Pi_{\parFb_{123},\parEb_{13},\parEb_1} -
\Pi_{\parFb_{123},\parFb_{13},\parEb_1}
\big)
+\frac {\sqrt{d_{\parEb_1,\parFb_{123}}^2-1}}
 {d_{\parEb_1,\parFb_{123}}^2}
\Pi_{\parFb_{123},\parFb_{13},\parEb_1\leftrightarrow\parFb_{123},\parEb_{13},\parEb_1},
\end{multline}
where
\[
\Pi_{\parFb_{123},\parFb_{13},\parEb_1\leftrightarrow\parFb_{123},\parEb_{13},\parEb_1} =
\sum\nolimits_{\parA>\parEb_1} \PiA_{\parFb_{123},\parFb_{13},\parEb_1\leftrightarrow\parFb_{123},\parEb_{13},\parEb_1}
\]
for short.

Without loss of generality, let us assume $N^{2/3}$ to be an integer. Then,
by using the branching rule and simple derivations, one can see that
\begin{equation} \label{eqn:SumMove}
\sum_{k=0}^{N^{2/3}-1}\frac{N^{2/3}-k}{N}
\sum_{\parE\vdash k}
\bigg(
\Pi_{\parEb_{123},\parEb_{1}}
+ \sum_{\parF<\parE}
\Pi_{\parFb_{123},\parFb_{13},\parEb_{1}}
\bigg)
=
 \sum_{k=0}^{N^{2/3}-1}
\bigg(\frac{1}{N}
\sum_{\parE\vdash k}
\Pi_{\parEb_{123},\parEb_{1}}
+
\frac{N^{2/3}-k}{N}
\sum_{\parF\vdash k-1}
\Pi_{\parFb_{123},\parFb_{13}}
\bigg).
\end{equation}
%
Therefore we have
\begin{align}
\notag
(\Gamma_\cA)_{1,2} & =
\sum_{k=0}^{N^{2/3}-1}\frac{N^{2/3}-k}{N}
\sum_{\parE\vdash k}
\Big(
\Pi_{\parEb_{123},\parEb_{1}} + \sum_{\parF<\parE}
\Pi_{\parFb_{123},\parEb_{12},\parEb_{1}}
\Big) \Pi_{id}
\\ & = \notag
\sum_{k=0}^{N^{2/3}-1}
\Bigg(
\frac{1}{N}
\sum_{\parE\vdash k}
\Pi_{\parEb_{123},\parEb_1}
+
\frac{N^{2/3}-k}{N}
\sum_{\parE\vdash k} \sum_{\parF<\parE}
\Big(
\frac {1} {d_{\parEb_1,\parFb_{123}}^2}
(\Pi_{\parFb_{123},\parEb_{13},\parEb_1}
- \Pi_{\parFb_{123},\parFb_{13},\parEb_1})
\\ & \hspace{70pt}
+
\frac {\sqrt{d_{\parEb_1,\parFb_{123}}^2-1}}{d_{\parEb_1,\parFb_{123}}^2}
\Pi_{\parFb_{123},\parFb_{13},\parEb_1 \leftrightarrow \parFb_{123},\parEb_{13},\parEb_1}
 \Big)
+
\frac{N^{2/3}-k}{N}
\sum_{\parF\vdash k-1}
\Pi_{\parFb_{123},\parFb_{13}}
 \Bigg)\Pi_{id},
\notag
\end{align}
where the first equality comes from the branching rule and the fact that we can ignore $k=N^{2/3}$,
and the second equality comes from subsequent applications of \refeqn{ThetaDec2} and \refeqn{SumMove}.

Recall that the action of $\Delta_3$ commutes with $\Pi_{id}$ and $\Delta_3\circ\Pi_{\parFb_{123},\parFb_{13}}=0$. Therefore we can choose
\begin{multline*}
\Delta_3\diamond(\Gamma_\cA)_{1,2} 
 =
\sum_{k=0}^{N^{2/3}-1}
\Bigg(
\frac{1}{N}
\sum_{\parE\vdash k}
\Pi_{\parEb_{123},\parEb_1}
+
\frac{N^{2/3}-k}{N}
\sum_{\parE\vdash k} \sum_{\parF<\parE}
\Big(
\frac {1} {d_{\parEb_1,\parFb_{123}}^2}
(\Pi_{\parFb_{123},\parEb_{13},\parEb_1}
- \Pi_{\parFb_{123},\parFb_{13},\parEb_1})
\\ 
+
\frac {\sqrt{d_{\parEb_1,\parFb_{123}}^2-1}}{d_{\parEb_1,\parFb_{123}}^2}
\Pi_{\parFb_{123},\parFb_{13},\parEb_1 \leftrightarrow \parFb_{123},\parEb_{13},\parEb_1}
 \Big)
 \Bigg)\Pi_{id},
\end{multline*}
and we have
\begin{align}
& (\Delta_3\diamond(\Gamma_\cA)_{1,2})^*
(\Delta_3\diamond(\Gamma_\cA)_{1,2}) 
\notag
\\ & \hspace{25pt} =
\sum_{k=0}^{N^{2/3}-1}
\Pi_{id}\Bigg(
\frac{1}{N^2}
\sum_{\parE\vdash k}
\Pi_{\parEb_{123},\parEb_1}
 +
\frac{(N^{2/3}-k)^2}{N^2}
\sum_{\parE\vdash k} \sum_{\parF<\parE}
\frac {1} {d_{\parEb_1,\parFb_{123}}^2}
\big(\Pi_{\parFb_{123},\parEb_{13},\parEb_1}
+ \Pi_{\parFb_{123},\parFb_{13},\parEb_1}\big)
 \Bigg)\Pi_{id}, \notag
\\ & \hspace{50pt} \preceq \frac{1}{N^2}
\sum_{k=0}^{N^{2/3}-1}
\Pi_{id}\Bigg(
\sum_{\parE\vdash k}
\Pi_{\parEb_{123},\parEb_1}
 +
o(1)\cdot
\sum_{\parE\vdash k} \sum_{\parF<\parE}
\big(\Pi_{\parFb_{123},\parEb_{13},\parEb_1}
+ \Pi_{\parFb_{123},\parFb_{13},\parEb_1}\big)
 \Bigg)\Pi_{id} \preceq \frac{1}{N^2}\I.
\notag
\end{align}
Finally, \refeqn{Half2Norm} tells us that
\[
\|\Delta_1\diamond \Gamma''_\cA\|^2 
=
\Big\|\sum_{\pi\in R''} V_{\pi} (\Delta_3\diamond(\Gamma_\cA)_{1,2})^*
(\Delta_3\diamond(\Gamma_\cA)_{1,2})  V_{\pi^{-1}}\Big\|
\leq \Big\|\sum_{\pi\in R''} \frac{1}{N^2}\I\,\Big\| \leq 1/2,
\]
and, hence, $\|\Delta_1\circ \Gamma''_\cA\|\in\O(1)$.

We show that $\|\Delta_1\circ \Gamma''_\cB\|\in\O(1)$ in essentially the same way, except now, instead of the decomposition \refeqn{ThetaDec2} of $\Pi_{\parFb_{123},\parEb_{12},\parEb_1}$ we consider the decomposition \refeqn{ThetaDec1} of $\PiEb_{\parFb_{123},\parEb_{12}}$. This concludes the proof that $\|\Delta_1\circ \Gamma''\|\in\O(1)$, which, in turn, concludes the proof of \refthm{main}.

\section{Open problems} \label{sec:open}

We already mentioned two open problems in the introduction. One is to close the gap between the best known lower bound and upper bound for {\sc $k$-Distinctness}, $\Omega(N^{2/3})$ and $\O(N^{1-2^{k-2}/(2^k-1)})$, respectively. We hope that our lower bound for {\sc Element Distinctness} could help to improve the lower bound for {\sc $k$-Distinctness} when $k\geq 3$.

The other is to reduce the required group (i.e., alphabet) size in the $\Omega(N^{k/(k+1)})$ lower bound for {\sc $k$-Sum}. As pointed out in \cite{spalek:kSumLower}, the quantum query complexity of {\sc $k$-Sum} becomes $\O(\sqrt{N})$ for groups of constant size. Therefore it would be interesting to find tradeoffs between the quantum query complexity and the size (and, potentially, the structure) of the group. These tradeoffs might be relatively smooth, unlike the jump in the query complexity of {\sc Element Distinctness} between alphabet sizes $N-1$ and $N$. 

Claims \ref{clm:GammaTransp} and \ref{clm:alphas} suggest that the adversary matrix that we  consider in \refthm{main} for {\sc Element Distinctness} is a natural choice.
While any other optimal adversary matrix probably cannot look too different (in terms of the singular value decomposition), it does not mean that it cannot have a simpler specification. Such a simpler specification might facilitate the construction of adversary bounds for other problems.

In fact, Belovs' construction \cite{belovs:adv-el-dist} gives an adversary matrix $\Gamma$ for {\sc Element Distinctness} for any alphabet size. Unfortunately, his analysis for lower bounding $\|\Gamma\|/\|\Delta_i\circ\Gamma\|$ does not work any more for alphabet sizes $o(N^2)$. Nonetheless, it still might be the case that $\|\Gamma\|/\|\Delta_i\circ\Gamma\|\in \Omega(N^{2/3})$ even when $|\Sigma|=N$, and, if one could show that, it might help to provide tight adversary bounds for {\sc Collision} and {\sc Set Equality} with minimal non-trivial alphabet size, because the current adversary  bounds for them are constructed similarly to Belovs's adversary bound for {\sc Element Distinctness} and require $|\Sigma|\in\Omega(N^2)$. (We know that such adversary bounds for {\sc Collision} and {\sc Set Equality} exist due to tight lower
bounds via other methods \cite{shi:collisionLower,kutin:collisionLower,zhandry:set-equal} and the optimality
of the adversary method~\cite{lee:stateConversion}.)

Jeffery, Magniez, and de Wolf recently studied the model of parallel quantum query algorithms, which can make $P$ queries in parallel in each timestep \cite{jeffery:parallel-complexity}. 
They show that such algorithms have to make $\Theta((N/P)^{2/3})$ $P$-parallel quantum queries to solve {\sc Element Distinctness}. For the lower bound, they generalize the adversary bound given in \cite{belovs:nonAdaptiveLG} (which is almost equivalent to one in \cite{belovs:adv-el-dist}) and therefore require that the alphabet size is at least $\Omega(N^2)$. The techniques provided in this paper might help to remove this requirement.

\section*{Acknowledgments}

The author would like to thank Andris Ambainis, Aleksandrs Belovs, Robin Kothari, Hari Krovi, Abel Molina, and John Watrous for fruitful discussions and useful comments.
A large portion of the research described in this paper was conducted during author's visit to the University of Latvia.
The author acknowledges the support of Mike and Ophelia Lazaridis Fellowship, David R. Cheriton Graduate Scholarship, and the US ARO.


\newcommand{\etalchar}[1]{$^{#1}$}

\appendix

\section{\texorpdfstring
{Necessary conditions for the construction of $\Gamma$}
{Necessary conditions for the construction of the adversary matrix}
}\label{app:necessary}

\subsection{\texorpdfstring
{Action of $\Delta_i$ on $\Pi^\parA_\parA$ and transporters}
{Action of the difference matrix on projectors and transporters}}

Let us consider $i\neq 2$.
Recall the projectors $\hat\Pi^s_i$ from \refsec{DeltaComm}, 
and note that $V_\pi^\tau\hat\Pi^s_i=\hat\Pi^s_i V_\pi^\tau$ for all $(\pi,\tau)\in\S_{[N]\setminus\{i\}}\times\S_{\Sigma\setminus\{s\}}$. Analogously to \refclm{VDecomp},
\[
\hat\Pi^s_i=\sum\nolimits_{\parB\vdash N-1}\hat\Pi^{s,\parB_s}_{i,\parB_i},
\]
where $\hat\Pi^{s,\parB_s}_{i,\parB_i}=\hat\Pi^s_i\Pi^{\parB_s}_{\parB_i}=\Pi^{\parB_s}_{\parB_i}\hat\Pi^s_i$ projects on a single instance of the irrep $\parB\times\parB$ of $\S_{[N]\setminus\{i\}}\times\S_{\Sigma\setminus\{s\}}$.

Due to the symmetry, $V_\pi^\tau(\overline\Delta_i\circ\PiA_\parA)=(\overline\Delta_i\circ\PiA_\parA) V_\pi^\tau$ for all $(\pi,\tau)\in\S_{[N]\setminus\{i\}}\times\S_\Sigma$, therefore we can express
\[
\overline\Delta_i\circ\PiA_\parA = \sum_{\parA'\vdash N}\sum_{\parB<\parA'} \phi^{\parA'}_\parB \Pi^{\parA'}_{\parB_i}.
\]
We have
\begin{multline*}
\phi^{\parA'}_\parB
= \frac{\Tr[(\overline\Delta_i\circ\PiA_\parA)\Pi^{\parA'}_{\parB_i}]}{\Tr[\Pi^{\parA'}_{\parB_i}]} 
= \frac{\Tr[\sum_{s\in\Sigma}\hat\Pi^s_i\PiA_\parA\hat\Pi^s_i\Pi^{\parA'}_{\parB_i}]}{\dim\parA'\dim\parB}
= N\frac{\Tr[\hat\Pi^{s,\parB_s}_{i,\parB_i}\PiA_{\parB_i}
                    \hat\Pi^{s,\parB_s}_{i,\parB_i} \Pi^{\parA'}_{\parB_i}]}
{\dim\parA'\dim\parB}
= N\frac{\Tr[\hat\Pi^{s,\parB_s}_{i,\parB_i}\PiA_{\parB_i}]
              \cdot\Tr[\hat\Pi^{s,\parB_s}_{i,\parB_i} \Pi^{\parA'}_{\parB_i}]}
{\dim\parA'(\dim\parB)^3}
\\
= N\frac{\Tr[\hat\Pi^s_i\PiA_{\parB_i}]
              \cdot\Tr[\hat\Pi^s_i \Pi^{\parA'}_{\parB_i}]}
{\dim\parA'(\dim\parB)^3}
= \frac{\Tr[\PiA_{\parB_i}]
              \cdot\Tr[\Pi^{\parA'}_{\parB_i}]}
{N\dim\parA'(\dim\parB)^3}
=
\begin{cases}
\frac{\dim\parA}{N\dim\parB}, & \text{if }\parB<\parA,
\\
0, & \text{if }\parB\not<\parA\text{ (i.e., $\PiA_{\parB_i}=0$)},
\end{cases}
\end{multline*}
where the second equality is due to \refeqn{DeltaViaProj}, 
the third and sixth equalities are due to the symmetry among all $s\in\Sigma$,
and the fourth equality is from \cite{ambainis:symmetryAssisted}.
Hence
\begin{equation}
\label{eqn:DeltaOnProj}
\Delta_i\circ\PiA_\parA
= \PiA_\parA-\frac{\dim\parA}{N}\sum_{\parB<\parA} 
\Big(
\frac{1}{\dim\parB} \sum_{\parA'>\parB} \Pi^{\parA'}_{\parB_i}
\Big)
= \PiA_\parA-\frac{\dim\parA}{N}\sum_{\parB<\parA} 
\Big(
\frac{1}{\dim\parB} \Pi_{\parB_i}
\Big).
\end{equation}

Now consider $j\neq i$, $\parA\vdash N$, and $\parC\ll_{rc}\parA$.
Let $\parB,\parB'\vdash N-1$ be such that $\parC<\parB<\parA$, $\parC<\parB'<\parA$, and $\parB\neq\parB'$.
Let us see how $\overline\Delta_i$ acts on the transporter $\Pi^\parA_{\parC_{ij},\parB'_i\leftarrow\parC_{ij},\parB_i}$. We have
\[
\hat\Pi^s_i  \Pi^\parA_{\parC_{ij},\parB'_i\leftarrow\parC_{ij},\parB_i}  \hat\Pi^s_i
=
\hat\Pi^s_i
\Pi^{\parB'_s}_{\parB'_i}
\Pi^\parA_{\parC_{ij},\parB'_i\leftarrow\parC_{ij},\parB_i}
\Pi^{\parB_s}_{\parB_i}
\hat\Pi^s_i = 0
\]
because $\Pi^{\parB'_s}\Pi^\parA_{\parC_{ij}, \parB'_i\leftarrow\parC_{ij},\parB_i}$ is a transporter between two instances of the irrep $\parC\times\parB'$ of $\S_{[N]\setminus\{i,j\}}\times\S_{\Sigma\setminus\{s\}}$ and, therefore, orthogonal to $\Pi^{\parB_s}$.
Hence,
\begin{equation}
\label{eqn:DeltaOnTrans}
\overline\Delta_i\circ \Pi^\parA_{\parC_{ij},\parB'_i\leftarrow\parC_{ij},\parB_i} = 0
\qqAnd
\Delta_i\circ \Pi^\parA_{\parC_{ij},\parB'_i\leftarrow\parC_{ij},\parB_i} =
\Pi^\parA_{\parC_{ij},\parB'_i\leftarrow\parC_{ij},\parB_i}.
\end{equation}

\subsection{\texorpdfstring
{Necessary conditions for $\|\Delta_1\circ\Gamma\|\in\O(1)$}
{Necessary conditions}}

We will use the following lemmas and corollaries in the proof of \refclm{alphas}. Let $\Gamma_{1,2}$ be given as in \refeqn{Gamma12Blocks}, and $\Gamma$ be obtained from $\Gamma_{1,2}$ via \refeqn{GammaViaBlocks}.

\begin{lem} \label{lem:helper1}
Consider $\parA\vdash N$, $\parB<\parA$, $\parB'<\parA$, and $\parC<\parB,\parB'$ (we allow $\parB=\parB'$ here). If $\|\Delta_1\circ\Gamma'\|\leq1$, then
\[
 \|\PiA_{\parC_{12},\parB_1}(\Delta_1\circ\Gamma_{1,2})\PiA_{\parC_{12},\parB'_1}\|
\leq\sqrt\frac{\dim\parB'}{(N-1)\dim\parC}.
\]
\end{lem}

\begin{proof}
For the proof, let us assume that $\parC\ll_{rc}\parA$ and $\parB\neq\parB'$. It is easy to see that the proof works in all the other cases too. 
Let
\(
\Psi_{\parC,\parB}^\parA=\sum_{\pi\in R'}U_\pi \PiA_{\parC_{12},\parB_1} U_{\pi^{-1}},
\)
where the transversal $R'$ was defined in \refsec{TwoHalves}.
From \refeqn{GammaHalves}, we have
\begin{equation}\label{eqn:Psi1}
\Psi^\parA_{\parC,\parB}(\Delta_1\circ\Gamma')
=
\sum_{\pi\in R'}U_\pi \PiA_{\parC_{12},\parB_1} (\Delta_1\circ\Gamma_{1,2}) V_{\pi^{-1}},
\end{equation}
whose norm is at most $1$ because $\Psi^\parA_{\parC,\parB}$ is a projector.

We can express
\[
\PiA_{\parC_{12},\parB_1}(\Delta_1\circ\Gamma_{1,2}) 
=
\psi \PiA_{\parC_{12},\parB_1}
+
\psi' \PiA_{\parC_{12},\parB_1\leftarrow\parC_{12},\parB'_1},
\]
where
\[
\psi =  \|\PiA_{\parC_{12},\parB_1}(\Delta_1\circ\Gamma_{1,2})\PiA_{\parC_{12},\parB_1}\|
\qqAnd
\psi' =  \|\PiA_{\parC_{12},\parB_1}(\Delta_1\circ\Gamma_{1,2})\PiA_{\parC_{12},\parB'_1}\|.
\]
Hence,
\begin{equation}\label{eqn:Psi2}
(\Delta_1\circ\Gamma_{1,2})^*\PiA_{\parC_{12},\parB_1}(\Delta_1\circ\Gamma_{1,2})
=
\psi^2 \PiA_{\parC_{12},\parB_1}
+
(\psi')^2 \PiA_{\parC_{12},\parB'_1}
+
\psi\psi' \PiA_{\parC_{12},\parB_1\leftrightarrow\parC_{12},\parB'_1}.
\end{equation}
From \refeqn{Psi1}, \refeqn{Psi2}, and \refeqn{PermSum1}, we get
\[
(\Delta_1\circ\Gamma')^*\Psi^\parA_{\parC,\parB}(\Delta_1\circ\Gamma')
=\psi^2(N-1)\frac{\dim\parC}{\dim\parB}\PiA_\parB
+(\psi')^2(N-1)\frac{\dim\parC}{\dim\parB'}\PiA_{\parB'}.
\]
The norm of this matrix is at most $1$, which completes the proof.
\end{proof}

\begin{cor} \label{cor:helper2}
Let $\parC\vdash N-2$, $\parB>\parC$, and $\parA,\parA'>\parB$. If $\|\Delta_1\circ\Gamma'\|\leq 1$, then 
\[
\bigg|
\frac{\Tr[\Pi^{\parA}_{\parC_{12},\parB_1}\Gamma_{1,2}]}{\dim\parA\dim\parC}
-
\frac{\Tr[\Pi^{\parA'}_{\parC_{12},\parB_1}\Gamma_{1,2}]}{\dim\parA'\dim\parC}
\bigg|
\leq 2 \sqrt\frac{\dim\parB}{(N-1)\dim\parC}.
\]
\end{cor}

\begin{proof}
From \reflem{helper1}, we have
\begin{multline*}
 \big\|\PiA_{\parC_{12},\parB_1}(\Delta_1\circ\Gamma_{1,2})\PiA_{\parC_{12},\parB_1}\big\|
=
\frac{\big|\Tr\big[\PiA_{\parC_{12},\parB_1}(\Delta_1\circ\Gamma_{1,2})\big]\big|}
{ \dim\parA\dim\parC }
= 
\frac{\big|\Tr\big[(\Delta_1\circ\PiA_{\parC_{12},\parB_1})\Gamma_{1,2}\big]\big|}
{ \dim\parA\dim\parC }
\\
=
\frac{\big|\Tr\big[(\PiA_{\parC_{12},\parB_1}-\frac{\dim\parA}{N\dim\parB}\Pi_{\parC_{12},\parB_1})\Gamma_{1,2}\big]\big|}
{ \dim\parA\dim\parC }
=
\bigg|
\frac{\Tr\big[\PiA_{\parC_{12},\parB_1}\Gamma_{1,2}\big]}{\dim\parA\dim\parC}
-\frac{\Tr\big[\Pi_{\parC_{12},\parB_1}\Gamma_{1,2}\big]}{N\dim\parB \dim\parC}
\bigg|
\leq \sqrt\frac{\dim\parB}{(N-1)\dim\parC},
\end{multline*}
where the second and third equalities are due to \refeqn{DeltaViaProj} and \refeqn{DeltaOnProj},
respectively. We obtain the same inequality with $\parA'$ instead of $\parA$, and the result follows from the triangle inequality.
\end{proof}

\begin{cor} \label{cor:helper3}
Consider $\parA\vdash N$, $\parC\ll_{rc}\parA$, and $\parB,\parB'\vdash N-1$ such that $\parC<\parB<\parA$, $\parC<\parB'<\parA$, and $\parB$ appears after $\parB'$ in the lexicographical order. If $\|\Delta_1\circ\Gamma'\|\leq1$, then
\[
\bigg|
\alpha_{id,\parC}^\parA
\frac{\sqrt{d_{\parA,\parC}^2-1}}{2d_{\parA,\parC}}
-
\alpha_{sgn,\parC}^\parA
\frac{d_{\parA,\parC}-1}{2d_{\parA,\parC}}
\bigg|
\leq\sqrt{\frac{\dim\parB}{(N-1)\dim\parC}},
\]
\[
\bigg|
\alpha_{id,\parC}^\parA
\frac{\sqrt{d_{\parA,\parC}^2-1}}{2d_{\parA,\parC}}
+
\alpha_{sgn,\parC}^\parA
\frac{d_{\parA,\parC}+1}{2d_{\parA,\parC}}
\bigg|
\leq\sqrt{\frac{\dim\parB'}{(N-1)\dim\parC}},
\]
\end{cor}

\begin{proof}
Since $\parA$ is the unique $N$-box Young diagram that has both $\parB$ and $\parB'$ as subdiagrams, we have
\[
\Pi_{\parC_{12},\parB'_{1}}\Gamma_{1,2} \Pi_{\parC_{12},\parB_{1}}
=
\PiA_{\parC_{12},\parB'_{1}}\Gamma_{1,2} \PiA_{\parC_{12},\parB_{1}}.
\]
Hence, due to \refeqn{DeltaOnTrans} and the commutativity relations \refeqn{DeltaComm}, we have
\[
\PiA_{\parC_{12},\parB'_{1}}(\Delta_1\circ\Gamma_{1,2}) \PiA_{\parC_{12},\parB_{1}}
=
\PiA(\Delta_1\circ(\Pi_{\parC_{12},\parB'_{1}}\Gamma_{1,2} \Pi_{\parC_{12},\parB_{1}}))\PiA
=
\PiA_{\parC_{12},\parB'_{1}}\Gamma_{1,2} \PiA_{\parC_{12},\parB_{1}}.
\]
The same holds with $\parB$ and $\parB'$ swapped. 
From \refeqn{PiId} and \refeqn{PiSgn}, we get that
\[
\PiA_{\parC_{12},\parB'_{1}}\Gamma_{1,2} \PiA_{\parC_{12},\parB_{1}}
=
\bigg(
\alpha_{id,\parC}^\parA
\frac{\sqrt{d_{\parA,\parC}^2-1}}{2d_{\parA,\parC}}
-
\alpha_{sgn,\parC}^\parA
\frac{d_{\parA,\parC}-1}{2d_{\parA,\parC}}
\bigg)
\PiA_{\parC_{12},\parB'_1\leftarrow\parC_{12},\parB_1},
\]
\[
\PiA_{\parC_{12},\parB_{1}}\Gamma_{1,2} \PiA_{\parC_{12},\parB'_{1}}
=
\bigg(
\alpha_{id,\parC}^\parA
\frac{\sqrt{d_{\parA,\parC}^2-1}}{2d_{\parA,\parC}}
+
\alpha_{sgn,\parC}^\parA
\frac{d_{\parA,\parC}+1}{2d_{\parA,\parC}}
\bigg)
\PiA_{\parC_{12},\parB_1\leftarrow\parC_{12},\parB'_1},
\]
and we apply \reflem{helper1} to complete the proof.
\end{proof}

\begin{lem} \label{lem:helper4}
Let $\parF$ be a Young diagram having  at most $N/2-2$ boxes and $\parE>\parF$. If $\|\Delta_1\circ\Gamma''\|\leq1$, then
\[
\Bigg|
\alpha^\parEB_{id,\parEB_{12}} -
\alpha^{\parFB}_{id,\parFB_{12}}
+
\frac
{2(\alpha^\parEB_{id,\parFB_{12}}-\alpha^\parEB_{id,\parEB_{12}})}
{d_{\parEB,\parFB_{12}}(d_{\parEB,\parFB_{12}}-1)}
\Bigg|
\leq 2\sqrt\frac{\dim\parFB_3}{\binom{N-1}{2}\dim\parFB_{123}}.
\]
\end{lem}
\begin{proof}
 Note that $\PiEb_{\parFB_{123},\parFB_3}(\Delta_3\circ\Gamma_{1,2})$ can be expressed as a linear combination of $\PiEb_{\parFB_{123},\parFB_3}$ and 
$\PiEb_{\parFB_{123},\parFB_3\leftarrow id,\parFB_{123},\parEB_3}$, while
$\Pi^{\parFB}_{\parFB_{123}}(\Delta_3\circ\Gamma_{1,2})$ is proportional to $\Pi^{\parFB}_{\parFB_{123}}$.
 Similarly to \reflem{helper1}, we can show that
\[
 \big\|
\PiEb_{\parFB_{123},\parFB_3}(\Delta_3\circ\Gamma_{1,2})
\PiEb_{\parFB_{123},\parFB_3}
\big\|
\leq\sqrt\frac{\dim\parFB_3}{\binom{N-1}{2}\dim\parFB_{123}}
\qqAnd
\big\|
\Pi^{\parFB}_{\parFB_{123}}(\Delta_3\circ\Gamma_{1,2})
\Pi^{\parFB}_{\parFB_{123}}
\big\|
\leq\sqrt\frac{\dim\parFB_3}{\binom{N-1}{2}\dim\parFB_{123}},
\]
where, instead of \refeqn{PermSum1}, we have to use (analogously proven)
\[
 \sum_{\pi\in R''}
V_{\pi} \PiEb_{\parFB_{123},\parFB_3} V_{\pi^{-1}}
=
\binom{N-1}{2}
\frac{\dim\parFB_{123}}
{\dim\parFB_{3}}
\PiEb_{\parFB_3}
\qqAnd
 \sum_{\pi\in R''}
V_{\pi} \PiEb_{id,\parFB_{123},\parEB_3\leftrightarrow\parFB_{123},\parFB_3} V_{\pi^{-1}}
= 0.
\]

Then, similarly to \refcor{helper2}, we get
\[
 \Bigg|
\frac{\Tr[\PiEb_{\parFB_{123},\parFB_3}\Gamma_{1,2}]}
{\dim\parEB\dim\parFB_{123}}
-
\frac{\Tr[\Pi^{\parFB}_{\parFB_{123}}\Gamma_{1,2}]}
{\dim\bar\parF\dim\parFB_{123}}
\Bigg|
\leq 2\sqrt\frac{\dim\parFB_3}{\binom{N-1}{2}\dim\parFB_{123}}.
\]
We conclude by noticing that
\[
\Pi^{\parFB}_{\parFB_{123}}\Gamma_{1,2}
=
\Pi^{\parFB}_{\parFB_{123}}
\big(
\alpha_{id,\parFB_{12}}^\parFB \Pi^\parFB_{\parFB_{12}}
\big)
=
\alpha_{id,\parFB_{12}}^\parFB \Pi^\parFB_{\parFB_{123}}
\]
and, due to \refeqn{overlap},
\begin{multline*}
\PiEb_{\parFB_{123},\parFB_3} \Gamma_{1,2} \,\PiEb_{\parFB_{123},\parFB_3}
=
\PiEb_{\parFB_{123},\parFB_3}
\big(
\alpha^\parEB_{id,\parEB_{12}} \PiEb_{\parEB_{12}}
+
\alpha^\parEB_{id,\parFB_{12}} \PiEb_{id,\parFB_{12}}
\big)
\PiEb_{\parFB_{123},\parFB_3}
\\
=
\bigg(
\Big(
1-\frac2{d_{\parEB,\parFB_{12}}(d_{\parEB,\parFB_{12}}-1)}
\Big)
\alpha^\parEB_{id,\parEB_{12}}
+
\frac2{d_{\parEB,\parFB_{12}}(d_{\parEB,\parFB_{12}}-1)}
\alpha^\parEB_{id,\parFB_{12}}
\bigg)
\PiEb_{\parFB_{123},\parFB_3}.
\end{multline*}

\end{proof}

\subsection{Proof of \refclm{alphas} }

We can assume that all the coefficients $\beta$ in the expression \refeqn{GammaBlocks} for $\Gamma$ are at most $N$, as $N$ is the trivial upper bound on the quantum query complexity of {\sc Element Distinctness}. That, in turn, means that we can assume that the coefficients $\alpha$ in Point \ref{pt1}, Point \ref{pt2}, and Point \ref{pt3} of \refclm{alphas} are, respectively, at most $\O(1)$, $\O(\sqrt N)$, and $\O(N)$. Let us prove sequentially every point of the claim.

\paragraph{Point \ref{pt1}.}

Consider $k\in\O(1)$, $\parF\vdash k$, and $\parE>\parF$, so $d_{\parEB,\parFB_{12}}=N-\O(1)$ and $\dim\parFB_{3}/\dim\parFB_{123}\in\Theta(1)$.
From \reflem{helper4}, we get that
\(
|  \alpha^\parEB_{id,\parEB_{12}}  -  \alpha^{\parFB}_{id,\parFB_{12}}  |  \in  \O(1/N),
\)
which proves that $\alpha^{\parFB}_{id,\parFB_{12}}=N^{-1/3}+\O(1/N)$ by the induction over $k$, where we take $\alpha^{(N)}_{id,(N-2)}=N^{-1/3}$ as the base case.

\paragraph{Point \ref{pt2}.}

Consider $\parF\vdash\O(1)$ and $\parE>\parF$, so $\dim\parFB_{1}/\dim\parFB_{12}\in\Theta(1)$. 
 From the first inequality of \refcor{helper3} (in which we choose $\parA=\parEB$ and $\parC=\parFb_{12}$, forcing $\parB=\parFb_{1}$), we get that
\(
| \alpha_{id,\parFB_{12}}^\parEB - \alpha_{sgn,\parFB_{12}}^\parEB | \in\O(1/\sqrt N).
\)
From \refcor{helper2} (in which we choose $\parC=\parFB_{12}$, $\parB=\parFB_{1}$, $\parA=\parFB$, and $\parA'=\parEB$), we get
\[
\bigg|
\frac{\Tr[\Pi^{\parFB}_{\parFB_{12}}\Gamma_{1,2}]}{\dim\parFB\dim\parFB_{12}}
-
\frac{\Tr[\Pi^{\parEB}_{\parFB_{12},\parFB_1}\Gamma_{1,2}]}{\dim\parEB\dim\parFB_{12}}
\bigg|
\in \O(1/\sqrt N),
\]
where we have 
\[
\Pi^{\parFB}_{\parFB_{12}}\Gamma_{1,2}
= \alpha^{\parFB}_{id,\parFB_{12}} \Pi^{\parFB}_{\parFB_{12}}
\qqAnd
\Pi^{\parEB}_{\parFB_{12},\parFB_1} \Gamma_{1,2}
\Pi^{\parEB}_{\parFB_{12},\parFB_1}
=
\bigg(
\alpha^\parEB_{id,\parFB_{12}}
\frac{d_{\parEB,\parFB_{12}}-1}{2d_{\parEB,\parFB_{12}}}
+
\alpha^\parEB_{sgn,\parFB_{12}}
\frac{\sqrt{d_{\parEB,\parFB_{12}}^2-1}}{2d_{\parEB,\parFB_{12}}}
\bigg)
\Pi^{\parEB}_{\parFB_{12},\parFB_1}
\]
from \refeqn{PiId} and \refeqn{PiSgn}.
Therefore,
\(
| 
\alpha^\parFB_{id,\parFB_{12}} - 
(\alpha_{id,\parFB_{12}}^\parEB + \alpha_{sgn,\parFB_{12}}^\parEB)/2 | \in\O(1/\sqrt N),
\)
which together with previously proven $\alpha^\parFB_{id,\parFB_{12}}=N^{-1/3}+\O(1/N)$ 
and
\(
| \alpha_{id,\parFB_{12}}^\parEB - \alpha_{sgn,\parFB_{12}}^\parEB | \in\O(1/\sqrt N)
\)
imply
$\alpha_{id,\parFB_{12}}^\parEB = N^{-1/3}+\O(1/\sqrt N)$ and
$\alpha_{sgn,\parFB_{12}}^\parEB = N^{-1/3}+\O(1/\sqrt N)$.

\paragraph{Point \ref{pt3}.}

Consider $\parA\vdash N$ and $\parC\ll_c\parA$ that is obtained from $\parA$ by removing two boxes in different columns below the first row. Let us consider two cases.

Case 1: $\parC\ll_{rc}\parA$. Let $\parB,\parB'\vdash N-1$ be such that $\parC<\parB<\parA$, $\parC<\parB'<\parA$, and $\parB\neq\parB'$. Since $d_{\parA,\parC}\geq 2$,  $\dim\parB/\dim\parC\in\Theta(N)$, and $\dim\parB'/\dim\parC\in\Theta(N)$, both inequalities of \refcor{helper3} together imply $\alpha^\parA_{id,\parC}=\O(1)$ and $\alpha^\parA_{sgn,\parC}=\O(1)$.

Case 2: $\parC\ll_c\parA$ and $\parC\not\!\ll_r\parA$ (i.e., $\parC$ is obtained from $\parA$ by removing two boxes in the same, but not the first, row). Let $\parB\vdash N-1$ be the unique Young diagram that satisfies $\parC<\parB<\parA$, and let $\parA'$ be obtained from $\parB$ by adding a box in the first row. 
For Point \ref{pt2} we already have shown that $\alpha^{\parA'}_{id,\parC}\in o(1)$ and $\alpha^{\parA'}_{sgn,\parC}\in o(1)$, so, from \refcor{helper2} and $\dim\parB/\dim\parC\in\Theta(N)$, we get that $\alpha^\parA_{id,\parC}\in \O(1)$.

\end{document}